\documentclass[journal, onecolumn]{IEEEtran}

\usepackage{setspace}
\usepackage{cite}
\usepackage{graphicx}
\usepackage{amsmath, amsthm, amssymb}
\usepackage{mathtools}
\usepackage[linesnumbered,ruled]{algorithm2e}
\usepackage{algorithmic, float}
\usepackage{subfigure}
\usepackage{array}
\usepackage[caption=false,font=footnotesize]{subfig}
\usepackage{booktabs}

\newtheorem{theorem}{Theorem}
\newtheorem{corollary}{Corollary}
\newtheorem{proposition}{Proposition}
\newtheorem{lemma}{Lemma}

\newtheorem{example}{Example}

\newcommand{\argmin}{\mathop{\rm argmin}}
\newcommand{\argmax}{\mathop{\rm argmax}}
\newcommand{\less}{\leqslant}
\newcommand{\gre}{\geqslant}
\newcommand{\what}{\widehat}
\newcommand{\wtilde}{\widetilde}
\newcommand{\floor}[1]{\lfloor #1 \rfloor}
\newcommand{\defn}{\ensuremath{\!:=}}
\newcommand{\real}{\ensuremath{\mathbb{R}}}
\newcommand{\Exs}{\ensuremath{\mathbb{E}}}
\newcommand{\xone}{\ensuremath{\widehat{x}^{(1)}}}
\newcommand{\xzero}{\ensuremath{\widehat{x}^{(0)}}}
\newcommand{\wone}{\ensuremath{\widehat{w}^{(1)}}}
\newcommand{\wzero}{\ensuremath{\widehat{w}^{(0)}}}

\DeclarePairedDelimiter\ceil{\lceil}{\rceil}

\begin{document}

\title{Adaptive and Oblivious Randomized Subspace Methods for High-Dimensional Optimization: Sharp Analysis and Lower Bounds}

\author{Jonathan~Lacotte, Mert~Pilanci\footnote{A preliminary version of this work was presented in part~\cite{lacottehighdimopt} at the Advances in Neural Information Processing Systems 32 (NeurIPS 2019).}
\thanks{J. Lacotte and M. Pilanci are with the Electrical Engineering Department at Stanford University.}
}

\maketitle

\begin{abstract}
    We propose novel randomized optimization methods for high-dimensional convex problems based on restrictions of variables to random subspaces. We consider oblivious and data-adaptive subspaces and study their approximation properties via convex duality and Fenchel conjugates. A suitable adaptive subspace can be generated by sampling a correlated random matrix whose second order statistics mirror the input data. We illustrate that the adaptive strategy can significantly outperform the standard oblivious sampling method, which is widely used in the recent literature. We show that the relative error of the randomized approximations can be tightly characterized in terms of the spectrum of the data matrix and Gaussian width of the dual tangent cone at optimum. We develop lower bounds for both optimization and statistical error measures based on concentration of measure and Fano's inequality. We then present the consequences of our theory with data matrices of varying spectral decay profiles. Experimental results show that the proposed approach enables significant speed ups in a wide variety of machine learning and optimization problems including logistic regression, kernel classification with random convolution layers and shallow neural networks with rectified linear units.
\end{abstract}


\ifCLASSOPTIONpeerreview
\begin{center} \bfseries EDICS Category: 3-BBND \end{center}
\fi
\IEEEpeerreviewmaketitle

\section{Introduction}
\IEEEPARstart{H}{igh}-dimensional optimization problems are becoming ever more common in applications such as computer vision, natural language processing, robotics, medicine, genomics, seismology or weather forecasting, where the volume of the data keeps increasing at a rapid rate. It is also standard practice to use high-dimensional representations of data measurements such as random Fourier features~\cite{rahimi2008random} or pre-trained neural networks' features~\cite{yosinski2014transferable, ghorbani2020neural}. In this work, we are interested in solving a convex optimization problem of the form
\begin{align}
\label{eqprimal}
    x^* \defn \argmin_{x \in \real^d} \left\{ F(x) \defn f(Ax) + \frac{\lambda}{2} \|x\|_2^2 \right\}\,,
\end{align}
where $A \in \real^{n \times d}$ is a data matrix and $f$ is a convex function. Such convex optimization problems are typically formulated to fit a linear prediction model, or, they may occur as the subroutine of an optimization method, e.g., proximal optimization~\cite{parikh2014proximal}; prox-linear method for convex-smooth composite objectives~\cite{drusvyatskiy2019efficiency}. Moreover, several standard non-convex neural network training problems can be equivalently stated as convex optimization problem in higher dimensions \cite{pilanci2020neural,ergen2020training,ergen2020convex,ergen2020convexdeep}. In the large-scale setting $d \gg 1$, random projections are an effective way of performing dimensionality reduction~\cite{vempala2005random, mahoney2011randomized, drineas2016randnla, avron2010blendenpik, meng2014lsrn, pilanci2016fast}, and the common practice is to employ oblivious sampling or sketching matrices, which are typically randomized and fixed ahead of time. Furthermore, sketches can be iteratively applied  \cite{pilanci2016iterative,lacotte2019faster,ozaslan2019iterative,pilanci2017newton} or averaged in independent trials \cite{bartan2020distributed,bartan2020distributed2,derezinski2020debiasing} to reduce the approximation error.  However, it is not clear whether one can do better by adapting the sketching matrices to data. In fact, \emph{we will show that adaptive sketching matrices can significantly improve the approximation quality of the optimal solution of a convex smooth optimization problem, and we will characterize the recovery error in terms of the smoothness of the objective function and the spectral decay of the data matrix. Furthermore, we establish lower-bounds on the performance of the oblivious sketch that exhibit its fundamental limitations.}. 

Although the oracle complexity of first-order optimization methods has the property of being dimension-free~\cite{bubeck2015convex}, the cost of forming gradients and manipulating data matrices may be computationally prohibitive in large-scale settings, let alone second-order methods involving hessian computations. As a result, many sketching-based algorithms~\cite{pilanci2017newton, qu2016sdna, doikov2018randomized, luo2016efficient, gower2019rsn} have been specifically designed to address the computational issues of the Newton method by reducing the cost of solving the linear Newton system. The Newton sketch algorithm~\cite{pilanci2017newton} addresses the case where $n \gg d$ and $f$ is separable, and solves an approximate Newton system based on a sketch of the data matrix. Stochastic Dual Newton Ascent~\cite{qu2016sdna} requires knowledge of a global upper bound on the Hessian and then solves an approximate Newton system using random principal sub-matrices of that upper bound. The Randomized Subspace Newton (RSN) method~\cite{gower2019rsn} uses, at each iteration, an approximate descent direction $S (S^\top H S)^\dagger S^\top g$ where $S^\top $ is a $m \times d$ random embedding with $m \ll d$ and $H$ and $g$ are respectively the Hessian and gradient at the current iterate. The Randomized Block Cubic Newton method~\cite{doikov2018randomized} addresses block-separable convex optimization problems and combines the ideas of randomized coordinate descent~\cite{nesterov2012efficiency} and cubic regularization~\cite{nesterov2006cubic}. 

In this work, we take a perspective which is agnostic to the optimization algorithm. Our goal is to formulate a low-dimensional optimization problem of the form
\begin{align}
\label{eqsketchedprimal}
    \alpha^* \in \argmin_{\alpha \in \real^m} \left\{F(S\alpha) = f(AS \alpha) + \frac{\lambda}{2} \|S\alpha\|_2^2 \right\}\,,
\end{align}
where $S$ is a $d \times m$ random embedding, and then to construct a (potentially nonlinear) recovery map $\varphi : \real^m \to \real^d$ such that $\what x \defn \varphi(\alpha^*)$ is a close approximation of $x^*$, as measured by the relative recovery error
\begin{align}
    \frac{\|\what x - x^*\|_2}{\|x^*\|_2}\,.
\end{align}
The low-dimensional formulation~\eqref{eqsketchedprimal} draws connections with the aforementioned randomized Newton methods: using the linear reconstruction map $\varphi(\alpha) = S \alpha$, the Newton method applied to~\eqref{eqsketchedprimal} yields the update $\alpha_{t+1} = \alpha_t - \eta_t (S^\top \nabla^2 F(S\alpha_t) S)^\dagger S^\top \nabla F(S\alpha_t)$, i.e., $S\alpha_{t+1} = S\alpha_t - \eta_t S(S^\top \nabla^2 F(S\alpha_t) S)^\dagger S^\top \nabla F(S\alpha_t)$ which is the update of the RSN method. However, our perspective is different, as one is free to use any optimization method to solve for $\alpha^*$ and any reconstruction map $\varphi$.

Our approach falls within the scope of \emph{random subspace optimization methods}, where we restrict the range of the optimization variable to a random lower-dimensional subspace. In the context of convex smooth optimization, the authors of~\cite{zhang2013recoveringCOLT, zhang2014recoveringIEEE} propose to project the $d$-dimensional features of the data matrix $A$ using an \emph{oblivious} random embedding $S \in \real^{d \times m}$, chosen independently of the data. Based on the solution $\alpha^*_\dagger$ of a low-dimensional optimization problem similar to~\eqref{eqsketchedprimal}, they compute the dual solution $y^* = \nabla f(AS\alpha^*_\dagger)$ and then set $\what x \defn -\lambda^{-1}A^\top y^*$. Interestingly, the resulting recovery error is smaller than the error of the linear estimator $S\alpha^*_\dagger$, although these guarantees hold under some arguably restrictive assumptions (e.g., low-rank data matrix $A$, or, $x^*$ approximately lies in the span of the top singular vectors of $A$). On the other hand, such a linear estimator has been considered in~\cite{xu2017efficient} in the context of generalization bounds for empirical risk minimization of convex Lipschitz loss functions. They use an \emph{adaptive} (data-dependent) random embedding $S$ of the form $S = A^\top \wtilde S$, where $\wtilde S$ is itself oblivious. The authors study the generalization error of the approximate solution, and relate it to the norm of the tail singular values of $A$. Our approach draws connections with these two works: similarly to~\cite{xu2017efficient}, we will consider adaptive random embeddings of the form $S = A^\top \wtilde S$; similarly to~\cite{zhang2013recoveringCOLT, zhang2014recoveringIEEE}, we will consider a non-linear dual mapping $\varphi$ and study the recovery error of $\varphi(\alpha^*)$. As for the intriguing connection between the RSN method and the linear recovery map $\varphi(\alpha) = S\alpha$, the dual recovery map has been shown~\cite{wang2017sketching} to be equivalent to the Hessian sketch~\cite{pilanci2015randomized} applied to the Fenchel dual program.

The dual recovery map has also been analyzed in the specific context of sparse recovery. The authors of~\cite{zhang2013recoveringCOLT, zhang2014recoveringIEEE} establish that, with an oblivious random embedding $S$, accurate recovery is guaranteed for a sketch size scaling at least as the sparsity level of $x^*$ (i.e., its number of non-zero entries), under the assumption that the support set of $x^*$ includes the most important coordinates of the data matrix (see Theorem~4 in~\cite{zhang2014recoveringIEEE}). In the context of support vector machines (SVM) classification where the dual solution may be sparse (as measured by the number of support vectors), the authors of~\cite{yang2015theory, yao2018high} propose to add a sparsity-promoting regularization term to the dual of the low-dimensional problem~\eqref{eqsketchedprimal}, and provide recovery guarantees in terms of the sparsity of the dual solution.

Besides accurate recovery in convex smooth optimization through a low-dimensional formulation, random subspace optimization with oblivious embeddings has been used and analyzed for SVM-classification and the preservation of margins~\cite{blum2005random, shi2012margin, paul2013random}, for large-scale trust region problems with oblivious embeddings~\cite{vu2017random}, for scaling up linear systems solvers~\cite{gower2015randomized} and for statistically optimal prediction through kernel ridge regression~\cite{yang2017randomized}. In the latter work, the authors approximate an empirical kernel matrix $K$ by sketching its columns with an oblivious embedding $\wtilde S$. For a kernel based on a finite-dimensional feature space, i.e., $K = \Phi \Phi^\top$, the sketch $K\wtilde S$ satisfies $K\wtilde S = \Phi S$ where $S = \Phi^\top \wtilde S$. Hence, this corresponds to a sketch of the data matrix $\Phi$ with the adaptive embedding $S$.

Adaptive embeddings of the form $S = A^\top \wtilde S$ (where $\wtilde S$ is itself oblivious) are reminiscent of randomized low-rank approximations methods~\cite{halko2011, boutsidis2013, witten2015randomized, derezinski2020precise}. More precisely, let us denote $P_S \defn S(S^\top S)^\dagger S^\top$ the linear projector onto the range of $S$ and $P_S^\perp \defn I-P_S$ the projector onto its orthogonal complement. Then, for commonly used embeddings $\wtilde S \in \real^{n \times m}$ such as Gaussian or the subsampled randomized Hadamard transform (SRHT), it holds with high probability that 
\begin{align}
    \|P_S^\perp A^\top\|_2 \lesssim \|A-A_m\|_2 + \frac{1}{\sqrt{m}}\|A-A_m\|_F\,,
\end{align} 
where $A_m$ is the best rank-$m$ approximation of $A$. Unsurprisingly, the aforementioned generalization error guarantees based on adaptive sketching proposed in~\cite{xu2017efficient} depend on $S$ through the critical quantity $\|P_S^\perp A^\top\|_2$, and so do the guarantees we develop in this manuscript. In other words, our randomized adaptive subspace optimization method is an approximate way of restricting the optimization variable to the span of the top $m$ singular vectors of $A$. Naturally, using the deterministic embedding $S=V_m$, where $V_m$ is the matrix of top $m$ right singular vectors of $A$, could yield stronger guarantees than our adaptive sketch method. However, there are strong computational benefits in using a randomized embedding~\cite{halko2011} instead of an exact singular value decomposition (SVD) algorithm, and the right combination of this randomized method along with the dual recovery map is not yet understood in the case of convex smooth optimization.

Our work relates to the considerable amount of literature on randomized approximations of high-dimensional kernel matrices $K$. A popular approach consists of building a low-rank factorization of the matrix $K$, using a random subset of its columns~\cite{williams2001using, smola2004tutorial, drineas2005nystrom, kumar2012sampling}. The so-called \emph{Nystrom method} has proven to be effective empirically~\cite{yang2012nystrom}, and many research efforts have been devoted to improving and analyzing the performance of its many variants (e.g., uniform column sub-sampling, leverage-score based sampling), especially in the context of regularized regression~\cite{bach2013sharp, derezinski2018reverse, derezinski2020determinantal}. For conciseness of this manuscript, we will not consider explicitly column subsampling matrices, although most of our results may extend to them.

\subsection{Notations and assumptions}

We work with a data matrix $A \in \real^{n \times d}$ and we refer to $n$ as the sample size and $d$ the features' dimension. We denote by $\rho \less \min\{n,d\}$ the rank of $A$, and by $\sigma_1 \gre \sigma_2 \gre \dots \gre \sigma_\rho > 0$ its non-zero singular values. Given an embedding $S \in \real^{d \times m}$, we define the linear projector $P_S$ onto the range of $S$ and the linear projector $P_S^\perp$ onto the orthogonal complement of the range of $S$ as
\begin{align}
    P_S \defn S (S^\top S)^\dagger S^\top\,,\qquad P_S^\perp \defn I_d - P_S\,,
\end{align}
where the superscript $^\dagger$ denotes the pseudo-inverse. For a matrix $M \in \real^{p \times q}$ with arbitrary dimensions $p$, $q \gre 1$, we denote by $\|M\|_2$ its operator norm (i.e., its largest singular value), and by $\|M\|_F$ its Frobenius norm. For a real number $\eta \gre 1$, we denote the $L_\eta$-norm of a vector $w \in \real^p$ as $\|w\|_\eta = \left( \sum_{i=1}^p |w_i|^\eta\right)^{\frac{1}{\eta}}$, and the corresponding unit $L_\eta$-ball as $\mathcal{B}_\eta^p \defn \{w \in \real^p \mid \|w\|_\eta \less 1\}$. We introduce several measures of the size of a set $T \subset \real^p$: its radius is $\text{rad}(T) \defn \sup_{t \in T} \|t\|_2$, its diameter is $\text{diam}(T) \defn \sup_{t, t^\prime \in T}\|t - t^\prime\|_2$, and its Gaussian width is 
\begin{align}
\label{eqndefngaussianwidth}
    \omega(T) \defn \mathbb{E}_g\!\left\{\sup_{t \in T}\,\langle t, g\rangle\right\}\,,
\end{align}
where $g \sim \mathcal{N}(0,I_p)$. 

We work with a real-valued objective function $f$ which is defined over $\real^n$. Unless stated otherwise, we will assume the function $f$ to be convex, differentiable and $\mu$-strongly smooth for some $\mu > 0$, i.e., 
\begin{align}
    \|\nabla f(y) - \nabla f(x)\|_2 \less \mu \cdot \|y-x\|_2
\end{align}
for any $x, y \in \real^n$. We introduce the convex (or Fenchel) conjugate of $f$, defined as 
\begin{align}
    f^*(z) \defn \sup_{w \in \real^n} \left\{\langle w, z \rangle - f(w)\right\}\,.
\end{align}
We recall a few results from convex analysis (we refer the reader to the books~\cite{boyd2004convex, rockafellar2015convex} for more background and details). The domain of $f^*$ is defined as $\text{dom}\,f^* \defn \{z \in \real^n \mid f^*(z) < +\infty\}$ and is a closed convex set. The function $f^*$ is convex. Smoothness of $f$ implies that $f^*$ is $(1/\mu)$-strongly convex, i.e., $\langle y-z, g_y - g_z \rangle \gre \frac{1}{\mu} \cdot \|y-z\|_2^2$ for any $y$, $z \in \text{dom}\, f^*$ and for any $g_y \in \partial f^*(y)$ and $g_z \in \partial f^*(z)$, where $\partial f^*(\cdot)$ denotes the subgradient operator of $f^*$. We denote by $\text{int dom}\,f^*$ the interior of the domain of $f^*$.

We work with an arbitrary regularization parameter $\lambda > 0$. For comparing our guarantees with other methods, we will typically assume that $\frac{\lambda}{\mu}$ is within the range of the eigenvalues of the matrix $AA^\top$, i.e., $\frac{\lambda}{\mu} \approx \sigma_K^2$ for some threshold rank $1 \less K \less \rho$. This corresponds to a common (or desirable) choice of $\lambda$ in practice, when the goal is either to improve the condition number of the primal program~\eqref{eqprimal} for numerical stability purposes (a.k.a.~Tikhonov regularization), or, to discard the effect of the small singular values of $A$ (i.e., noise) for statistical estimation (a.k.a.~$\ell_2$-shrinkage).

For arbitrary dimensions $p \less q$, we say that $U \in \real^{p \times q}$ is a (partial) Haar matrix in $\real^{q}$ if $U U^\top = I_p$ and the range of $U$ is uniformly distributed among the $p$-dimensional subspaces of $\real^q$.

\subsection{Randomized sketches}

Given a dimension $p \in \{n,d\}$ and a \emph{sketch size} $m \less p$, we work with two oblivious classes of $m$-dimensional random embeddings $S \in \real^{p \times m}$, namely, Gaussian embeddings with independent and identically distributed (i.i.d.)~Gaussian entries $S_{ij} \sim \mathcal{N}(0,1/m)$, and, the SRHT~\cite{ailon2006approximate, sarlos2006improved, dobriban2019asymptotics, lacotte2020srht}, defined as $S = \sqrt{\frac{p}{m}} \cdot D \cdot H\cdot R$ where $R \in \real^{p \times m}$ is a column subsampling matrix, $D \in \real^{p \times p}$ is a diagonal matrix with independent entries uniformly sampled from $\{\pm 1\}$ and $H \in \real^{p \times p}$ is the Walsh-Hadamard transform. The $k$-th Walsh-Hadamard transform $H \equiv H_k$ is obtained by the recursion $H_0 = [1]$ and $H_{k+1} = \begin{bmatrix} H_k & H_k \\ H_k & -H_k \end{bmatrix}$, so that the dimension of $H_k$ is $2^k \times 2^k$. This requires the dimension $p$ to be a power of $2$. If not, a standard practice for sketching a matrix $M$ with $p$ columns is to form the matrix $\wtilde M = [M, 0]$ which has an additional number of $\wtilde p - p$ columns filled with zeros, where $\wtilde p$ is the smallest power of $2$ greater than $p$ (note that $\wtilde p \less 2p$), and to use the sketch $\wtilde M S$ with a SRHT $S \in \real^{\wtilde p \times m}$. For conciseness, in our formal statements involving a SRHT, we will implicitly assume that the relevant dimension $p$ is a power of $2$, although this does not restrict the applicability of our results thanks to the aforementioned zero-padding trick. In contrast to a Gaussian embedding, the SRHT verifies the orthogonality property $S^\top S = \frac{p}{m}\cdot I_m$. Furthermore, the recursive structure of the Walsh-Hadamard transform enables fast computation of a sketch $MS$, in time $\mathcal{O}(\text{nnz}(M) \log m)$ where $\text{nnz}(M)$ is the number of non-zero entries of $M$ (see, for instance,~\cite{boutsidis2013} for details), as opposed to a Gaussian embedding which requires time $\mathcal{O}(\text{nnz}(M)m)$ (using classical matrix multiplication).

\subsection{Oblivious vs Adaptive Sketches}

Data adaptive sketches of the form $S = A^\top \wtilde S$, where $\wtilde S$ is an oblivious sketching matrix, aim to mirror the correlation structure of the original data matrix $A$. For instance, for a Gaussian embedding $\wtilde S$, we observe that the columns of $S=A^\top \wtilde S$ are distributed as $S_i\sim \mathcal{N}(0,A^TA)$. Furthermore, in many statistical applications, assuming that the rows of $A$ are independent (and normalized) data sample vectors in $\real^d$, the matrix $A^T A$ corresponds to the $d\times d$ empirical covariance matrix. Therefore, one can expect the adaptive sketch to provide a more faithful summary of the data. However, this increased accuracy comes with the price of an additional pass over the dataset to compute the product $A^\top \wtilde S$, which is typically negligible.

\section{An overview of our contributions}
\label{sectionpreliminaries}

We introduce the Fenchel dual of the primal program~\eqref{eqprimal},
\begin{align}
\label{eqdual}
    z^* \in \argmin_{z \in \real^n} \left\{ f^*(z) + \frac{1}{2\lambda} \|A^\top z\|_2^2\right\}\,,
\end{align}
and the Fenchel dual of the sketched primal program~\eqref{eqsketchedprimal},
\begin{align}
\label{eqsketcheddual}
    y^* \in \argmin_{y \in \real^n} \left\{f^*(y) + \frac{1}{2\lambda}\|P_S A^\top y\|_2^2\right\}\,.
\end{align}
\begin{proposition}[Strong Fenchel duality]
\label{propositionfenchelduality}
There exist a unique \emph{primal} solution $x^* \in \real^d$ to~\eqref{eqprimal} and a unique \emph{dual} solution $z^* \in \textrm{dom}\,f^*$ to~\eqref{eqdual}, and these solutions are related through the Karush-Kuhn-Tucker (KKT) conditions $x^* = -\frac{A^\top z^*}{\lambda}$ and $z^* = \nabla f(Ax^*)$. If the function $f$ is strictly convex, then $z^* \in \text{int dom}\,f^*$.
\end{proposition}
\begin{proposition}[Strong Fenchel duality on sketched program]
\label{propositionsketchedfenchelduality}
There exist a sketched primal solution $\alpha^* \in \real^m$ to~\eqref{eqsketchedprimal} and a sketched dual solution $y^* \in \textrm{dom}\,f^*$ to~\eqref{eqsketcheddual}, and these solutions are related through the KKT conditions $S^\top S\alpha^* = -\frac{S^\top A^\top y^*}{\lambda}$ and $y^* = \nabla f(AS\alpha^*)$. If the function $f$ is strictly convex, then $y^* \in \text{int dom}\,f^*$.
\end{proposition}
Strong Fenchel duality is critical to understand the influence of the right-sketch $AS$ on the high-dimensional problem \eqref{eqprimal}. By duality, the right-sketch is identical to employing a left-sketch $P_S A^T$ for the dual problem \eqref{eqsketcheddual} where $P_S=S(S^TS)^\dagger S^T$ is the range space projector of $S$. As it will be shown, the sketch is significantly more accurate when the adaptive embedding $S=A^\top \wtilde S$ is employed. 

\subsection{Low-dimensional estimators}

Based on a low-dimensional solution $\alpha^*$ to~\eqref{eqsketchedprimal}, we focus on two candidate estimators of $x^*$. The most natural candidate is given by the linear mapping
\begin{align}
    \xzero \defn S\alpha^*\,,
\end{align}
which we refer to as the \emph{zero-order} estimator. On the other hand, the first-order optimality conditions $x^* = G_\lambda(x^*)$ where $G_\lambda(x) \defn -\frac{1}{\lambda} A^\top \nabla f(A x)$ suggest the estimator
\begin{align}
    \xone \defn G_\lambda(S\alpha^*)
\end{align}
and we refer to it as the \emph{first-order} estimator of $x^*$. We note that $\xone = G_\lambda(\xzero) = \xzero - \frac{1}{\lambda}\nabla F(\xzero)$, i.e., the first-order estimator $\xone$ is the result of applying a gradient step correction to $\xzero$ with step size $1/\lambda$. Moreover, it holds that
\begin{align}
    x^* = -\lambda^{-1} A^\top z^*\,,\qquad \xone = -\lambda^{-1} A^\top y^*\,.
\end{align}
In contrast to $\xzero$, the first-order estimator $\xone$ is the result of a linear mapping of the sketched dual solution $y^*$. Except for a quadratic function $f$, this corresponds to a non-linear transformation of $\alpha^*$. The estimator $\xzero$ is computed solely based on the sketched data matrix $AS$, whereas $\xone$ uses an additional call to $A$ through the mapping $G_\lambda$. These observations would suggest a better performance for $\xone$.

\subsection{Main contribution}

Our analysis involves a canonical geometric object in convex analysis, namely, the \emph{tangent cone} of the domain of $f^*$ at the dual solution $z^*$, defined as
\begin{align}
    \mathcal{T}_{z^*} \defn \{t \cdot (y - z^*) \mid t \gre 0,\,\,y \in \text{dom} f^*\}\,.
\end{align}
The tangent cone is the intersection of the supporting hyperplanes of the domain of $f^*$ at $z^*$. A critical quantity in our error guarantees is the maximal singular value $\mathcal{Z}_f$ of the matrix $P_S^\perp A^\top$ restricted to the spherical cap
\begin{align}
\label{eqndefnsphericalcap}
    \mathcal{C}_{z^*} \defn \mathcal{T}_{z^*} \cap \mathcal{B}_2^n
\end{align}
and formally defined as
\begin{align}
\label{eqndefnZf}
    \mathcal{Z}_f \defn \sup_{\Delta \in \mathcal{C}_{z^*}} \,\|P_S^\perp A^\top \Delta\|_2\,.
\end{align}
Our main result is the following deterministic upper bound on the relative recovery error of the first-order estimator $\xone$, which shows in particular that $\xone$ has better performance than $\xzero$, at least by a factor one half.
\begin{theorem}[Recovery error of the first-order estimator]
\label{theoremzeroversusfirstorder}
Let $S \in \real^{d \times m}$ be an embedding matrix, and let $\alpha^*$ be a minimizer of the sketched program~\eqref{eqsketchedprimal}. Under the condition $\lambda \gre 2 \mu \mathcal{Z}_f^2$, it holds that
\begin{align}
\label{eqnboundfirstorder}
    &\frac{{\|\xone-x^*\|}_2}{\|x^*\|_2} \less \sqrt{\frac{\mu}{2 \lambda}} \cdot \mathcal{Z}_f \cdot \min\!\left\{1, \frac{{\|\xzero-x^*\|}_2}{\|x^*\|_2}\right\}\,.
\end{align}
\end{theorem}
Naturally, we have $\sup_{\Delta \in \mathcal{T}_{z^*} \cap \mathcal{B}_2^n} \|P_S^\perp A^\top \Delta \|_2 \less \sup_{\Delta \in \mathcal{B}_2^n} \|P_S^\perp A^\top \Delta \|_2$, i.e., $\mathcal{Z}_f \less \|P_S^\perp A^\top\|_2$. When $z^*$ is an extreme point of the tangent cone $\mathcal{T}_{z^*}$, then one may expect $\mathcal{T}_{z^*}$ to have a small size and $\mathcal{Z}_f \ll \|P_S^\perp A^\top\|_2$. We discuss such instances in Section~\ref{sectionadaptive}.
On the other hand, when $z^*$ belongs to the interior of the domain of $f^*$, then $\mathcal{C}_{z^*} = \mathcal{B}_2^n$ so that $\mathcal{Z}_f = \|P_S^\perp A^\top\|_2$. For an adaptive embedding $S = A^\top \wtilde S$, we can then leverage well-known upper bounds~\cite{halko2011, boutsidis2013} on the residual error $\|P_S^\perp A^\top\|_2$ of the form
\begin{align}
    \|P_S^\perp A^\top\|_2 \lesssim R_{m/2}(A)\,,
\end{align}
where the \emph{spectral residual} $R_\delta(A)$ of the matrix $A$ at level $\delta > 0$ is defined as
\begin{align}
\label{eqnspectralresidual}
    R_\delta(A) \defn \sigma_{\floor{\delta}+1} + \frac{1}{\sqrt{\floor{\delta}}} \cdot \sqrt{\sum_{j=\floor{\delta}+1}^\rho \sigma_j^2}\,.
\end{align}
\noindent As an immediate consequence, we establish in Theorem~\ref{theoremboundfirstorderrandom} in Section~\ref{sectionadaptive} high-probability upper bounds on the recovery error,
\begin{align}
    \frac{\|\xone - x^*\|_2}{\|x^*\|_2} \lesssim \sqrt{\frac{\mu}{\lambda}} \cdot R_{m/2}(A)\,,
\end{align}
provided that $\lambda \gtrsim \mu R_{m/2}^2(A)$.

\subsection{Comparison to existing work and oblivious sketching}

The recovery method most related to ours is based on the first-order estimator $\xone_\dagger = G_\lambda(Q \alpha^*_\dagger)$ introduced in~\cite{zhang2013recoveringCOLT, zhang2014recoveringIEEE}, where $\alpha_\dagger^* \defn \argmin_{\alpha \in \real^m} \{f(AQ\alpha) + \frac{\lambda}{2}\|\alpha\|_2^2\}$ and $Q \in \real^{d \times m}$ is an oblivious embedding with i.i.d.~Gaussian entries $\mathcal{N}(0,1/m)$. In contrast to~\eqref{eqsketchedprimal}, the regularization term does not involve the sketching matrix $Q$, and this incurs performance guarantees different from ours: under several stringent assumptions (see Theorem~\ref{theoremzhangrecovering} in Section~\ref{sectionoblivious} for details), a \emph{best-case} upper bound on the recovery error of $\xone_\dagger$ is given as
\begin{align}
    \frac{\|\xone_\dagger - x^*\|_2}{\|x^*\|_2} \lesssim \sqrt{\frac{d_{\lambda/\mu}}{m}}\,,
\end{align}
where the \emph{effective dimension} $d_{\lambda/\mu}$ -- a critical quantity in many contexts for sketching-based algorithms~\cite{alaoui2015fast, chowdhury2018iterative, lacotte2020effective} -- is defined as
\begin{align}
    d_{\lambda/\mu} \defn \frac{\text{trace}(D_{\lambda/\mu})}{\|D_{\lambda/\mu}\|_2}\,,
\end{align}
and $D_{\lambda/\mu} \defn A (\frac{\lambda}{\mu} I_d + A^\top A)^{-1} A^\top$. The characteristic measure of error for our adaptive estimator $\xone$ is $\sqrt{\frac{\mu}{\lambda}} \cdot R_{m/2}(A)$ as opposed to $\sqrt{\frac{d_{\lambda/\mu}}{m}}$ for the oblivious estimator $\xone_\dagger$. In Table~\ref{tablecomparisonspectralresidualeffdim}, we compare these two measures for spectral decays which are common~\cite{gu2013smoothing, yang2017randomized} in machine learning, e.g., polynomial $\sigma_j \asymp j^{-\frac{1+\nu}{2}}$ or exponential $\sigma_j \asymp e^{-\frac{\nu j}{2}}$, for some $\nu > 0$ (the proofs of these results involve simple summations and we leave details to the reader). The ratio $\sqrt{\frac{d_{\lambda/\mu}}{m}}$ scales in terms of $m$ as $\mathcal{O}(m^{-\frac{1}{2}})$ regardless of the spectral decay, whereas the term $\sqrt{\frac{\mu}{\lambda}} \cdot R_{m/2}(A)$ has the better scaling $\mathcal{O}(m^{-\frac{1+\nu}{2}})$ for a polynomial decay and $\mathcal{O}(e^{-\frac{\nu m}{2}})$ for an exponential decay.
\begin{table}[!h]
\caption{Comparison of the characteristic measures of errors for polynomial and exponential spectral decays. We assume that $\frac{\lambda}{\mu} \asymp \sigma_K^2$ for some threshold rank $1 \less K \less \rho$, and that $m \gre 2 K$.}
  \label{tablecomparisonspectralresidualeffdim}
  \centering
  \begin{small}
  \begin{tabular}{|c|c|c|c|c|}
    \cmidrule(r){1-5}
    Estimator & Embedding & Characteristic error & Polynomial $\sigma_j \asymp j^{-\frac{1+\nu}{2}}$ & Exponential $\sigma_j \asymp e^{-\frac{-\nu j}{2}}$  \\
    \midrule
    $\xone$ (Ours) & Adaptive, Gaussian & $\sqrt{\frac{\mu}{\lambda}} R_{\frac{m}{2}}(A)$ & $\left(\frac{2K}{m}\right)^{\frac{1+\nu}{2}}$ & $e^{-\nu (\frac{m}{2}-K)}$ \\
    $\xone_\dagger$ (\!\cite{zhang2013recoveringCOLT, zhang2014recoveringIEEE}) & Oblivious, Gaussian &  $\sqrt{\frac{d_{\lambda/\mu}}{m}}$ & $\sqrt{\frac{K}{m}}$ & $\sqrt{\frac{K}{m}}$ \\
    \bottomrule
    \end{tabular}
    \end{small}
\end{table}

\subsection{Statement of additional contributions}

In addition to our main result (Theorems~\ref{theoremzeroversusfirstorder} and~\ref{theoremboundfirstorderrandom}) on the recovery error of $\xone$, we have the following contributions. In Section~\ref{sectionadaptive}, by leveraging the equivalence between adaptive sketching of the data matrix $A$ and oblivious sketching of the Gram matrix $AA^\top$, we extend our results to kernel methods (Theorem~\ref{thmainkernel}). In Section~\ref{sectionoblivious}, we establish lower bounds (Theorems~\ref{theoremboundzeroorderrandom} and~\ref{theoremlowerboundfirstorderobliviousrandom}) on the recovery error of the estimators $\xzero$ and $\xone$ with oblivious embeddings. We show that these recovery errors are bounded away from $0$ unless $m \approx d$, which would defeat the purpose of sketching. In Section~\ref{sectionalgorithms}, we provide a prototype algorithm for adaptive sketching (Algorithm~\ref{alg:MainAlgorithm}), and we show that our proposed low-dimensional formulation is at least as numerically stable as the original program~\eqref{eqprimal}. Furthermore, we extend Algorithm~\ref{alg:MainAlgorithm} to an iterative version (Algorithm~\ref{alg:Algorithm}) with the following guarantee (see Theorem~\ref{theoremiterativemethod}). Based on a single sketch $AS$, it returns after $T$ iterations a solution $\xone_T$ which satisfies with high probability
\begin{align}
    \frac{{\|\xone_T - x^*\|}_2}{{\|x^*\|}_2} \less \left(\frac{\mu \mathcal{Z}_f^2}{2\lambda}\right)^{\frac{T}{2}}\,,
\end{align}
provided that $\lambda \gre 2\mu \mathcal{Z}_f^2$. Consequently, one can construct an approximation $\xone_T$ with linear convergence rate $\sqrt{\frac{\mu \mathcal{Z}_f^2}{2\lambda}}$ based on a single sketch of the data matrix $A$, a number $T$ of matrix-vector multiplications of the form $A^\top \nabla f(\cdot)$ (which is equivalent to a gradient call to the function $F$), and through solving a number $T$ of low-dimensional optimization programs. In Section~\ref{sectionnonsmooth}, we extend our analysis to the case of a non-smooth objective function $f$. We show how to construct a first-order estimator $\xone$ based on a low-dimensional optimization program and which satisfies the following high-probability guarantee (see Theorem~\ref{theoremnonsmooth}),
\begin{align}
    \|\xone - x^*\|_2 \lesssim \frac{1}{\lambda} \cdot \mathcal{Z}_f\,.
\end{align}
Lastly, we show in Theorem~\ref{theoremlowerboundleastsquares} and Corollary~\ref{corollaryoptimalityzeroorder} that, as $\lambda \to 0$, the zero-order estimator $\xzero$ with an oblivious embedding achieves the minimax rate of optimality for a fundamental statistical problem, namely, estimating the mean of a Gaussian distribution $\mathcal{N}(Ax_\text{pl}, \frac{\sigma^2}{n}I_n)$ under the smoothness assumption $\|x_\text{pl}\|_2 \less 1$. In contrast to our main results, this suggests different benefits of the linear reconstruction map and oblivious embeddings for small values of the regularization parameter $\lambda$.

\section{Smooth, convex optimization in adaptive random subspaces}
\label{sectionadaptive}

\subsection{Restricted singular value with adaptive random embeddings}

According to Theorem~\ref{theoremzeroversusfirstorder}, the relative recovery error depends on the restricted singular value $\mathcal{Z}_f$, and we provide next an upper bound on $\mathcal{Z}_f$ in the case of an adaptive Gaussian embedding. Let $A = U \Sigma V^\top $ be a thin a singular value decomposition of $A$, where $U \in \real^{n \times \rho}$ and $V \in \real^{d \times \rho}$ have orthonormal columns, and $\Sigma \in \real^{\rho \times \rho}$ is the diagonal matrix of the non-zero singular values of $A$ in non-increasing order. For a given target rank $1 \less k \less \frac{\rho}{2}$, we define $\Sigma_k \defn \text{diag}\{\sigma_1, \hdots, \sigma_k\}$, $\Sigma_{\rho-k} \defn \text{diag}\{\sigma_{k+1}, \hdots, \sigma_{\rho}\}$, the matrix $U_k \in \mathbb{R}^{n \times k}$ with the first $k$ columns of $U$ and $U_{\rho-k} \in \mathbb{R}^{n \times (\rho-k)}$ with its last $(\rho-k)$ columns.
\begin{lemma}[Restricted singular value with adaptive embeddings]
\label{lemmaupperboundZf}
Let $S = A^\top \wtilde S$ where $\wtilde S \in \real^{n \times m}$ has i.i.d.~Gaussian entries and $m=2k$ for some target rank $1 \less k \less \frac{\rho}{2}$. Then, it holds that
\begin{align}
\label{upperboundZf}
    \mathcal{Z}_f \less c^\prime_g \cdot \textrm{rad}(U_k^\top \mathcal{C}_{z^*}) \cdot R_k(A) + \sup_{\Delta \in \mathcal{C}_{z^*}} {\|\Sigma_{\rho-k}U_{\rho-k}^\top \Delta\|}_2\,,
\end{align}
with probability at least $1-6 e^{-k}$, where $c^\prime_g$ is a universal constant which satisfies $c^\prime_g \less 25$. This implies in particular that, with probability at least $1-6e^{-k}$,
\begin{align}
\label{eqnupperboundresidualgaussian}
    \mathcal{Z}_f \less \|P_S^\perp A^\top\|_2 \less c_g \cdot R_k(A)\,,
\end{align}
where $c_g$ is a universal constant which satisfies $c_g \less 26$.
\end{lemma}
The upper bound~\eqref{upperboundZf} on $\mathcal{Z}_f$ involves the quantities $\textrm{rad}(U_k^\top \mathcal{C}_{z^*})$ and $\sup_{\Delta \in \mathcal{C}_{z^*}} {\|\Sigma_{\rho-k}U_{\rho-k}^\top \Delta\|}_2$ which are deterministic (they do not depend on the randomness of $S$). These quantities depend on the coupling between the matrix $U$ and the spherical cap $\mathcal{C}_{z^*}$ which may vary based on the problem at hand. Under some (idealized) assumptions, we are able to characterize their typical values, and thus a bound on $\mathcal{Z}_f$ which decouples $A$ and the spherical cap $\mathcal{C}_{z^*}$.
\begin{lemma}
\label{lemmarandomizedupperboundZd}
We assume the matrix $U \in \real^{n \times \rho}$ of left singular vectors of $A$ to be a random Haar matrix in $\real^n$, and the dual solution $z^*$ to be independent of $U$. Under the hypotheses of Lemma~\ref{lemmaupperboundZf}, there exists universal constants $c_0, c_1, c_2 > 0$ such that 
\begin{align}
     \mathcal{Z}_f \less c_0 \cdot \frac{\omega(\mathcal{C}_{z^*}) + \sqrt{m}}{\sqrt{n}} \cdot R_{m/2}(A)\,,
\end{align}
with probability at least $1-c_1 e^{-c_2m}$.
\end{lemma}
The upper bound $\frac{\omega(\mathcal{C}_{z^*}) + \sqrt{m}}{\sqrt{n}}$ on the ratio $\frac{\mathcal{Z}_f}{R_{m/2}(A)}$ decreases down to $\frac{\omega(\mathcal{C}_{z^*})}{\sqrt{n}}$ as $m$ decreases down to $\omega^2(\mathcal{C}_{z^*})$, and then plateaus at $\frac{\omega(\mathcal{C}_{z^*})}{\sqrt{n}}$. The smaller the Gaussian width $\omega(\mathcal{C}_{z^*})$, the more favorable the statistical-computational trade-off in terms of the sketch size. For instance, suppose that the domain of $f^*$ is the $L_1$-ball (e.g., $\ell_\infty$-regression; see Example~\ref{examplelinfty}), and denote by $s^*$ the number of non-zero entries of the dual solution $z^*$. It is known (see, for instance, Section~2.2.2 in~\cite{pilanci2016fast}) that $\mathcal{C}_{z^*} \subseteq \left\{ \Delta \in \real^n \mid \|\Delta\|_1 \less 2 \sqrt{s^*} \right\}$, and thus, $\omega(\mathcal{C}_{z^*}) \lesssim \sqrt{s^* \cdot \log n}$. Consequently,
\begin{align}
    \mathcal{Z}_f \lesssim \frac{\sqrt{s^* \cdot \log n} + \sqrt{m}}{\sqrt{n}} \cdot R_{m/2}(A)\,.  
\end{align}
The inequality $\|P_S^\perp A^\top\| \less c_g R_k(A)$ in~\eqref{eqnupperboundresidualgaussian} is well-known: it is in fact a simplified statement of the result of Corollary~10.9 in~\cite{halko2011}. We do not aim to derive a high-probability bound on $\mathcal{Z}_f$ with the SRHT as this would involve a different machinery of technical arguments, but a high-probability bound on $\|P_SA^\top\|_2$ already exists in the literature. 
\begin{lemma}[SRHT spectral residual~\cite{boutsidis2013}]
\label{lemmaresidualsrht}
Let $k \gre 2$ be a target rank and pick the sketch size $m \defn 19 (\sqrt{k} + 4 \sqrt{\log n})^2 \log(kn)$. Let $S = A^\top \wtilde S$ where $\wtilde S \in \real^{n \times m}$ is a SRHT. Then, it holds with probability at least $1-\frac{5}{n}$ that ${\|P_S^\perp A^\top\|}_2 \less c_s \cdot R_k(A)$ for some universal constant $c_s \less 5$.
\end{lemma}
Gaussian width based results established in this section complement the analysis of left-sketching in terms of the Gaussian width of the tangent cone \cite{pilanci2015randomized}. However, the results are quite different due to the extra regularization terms, the projection matrix $P_S^\perp$ appearing in the dual problem  and the effect of the spectral residual term $R_k(A)$.

\subsection{Recovery error in randomized adaptive random subspaces}

Combining the upper bound~\eqref{eqnboundfirstorder} on the recovery error of $\xone$ along with the results of Lemmas~\ref{lemmaupperboundZf} and~\ref{lemmaresidualsrht}, we obtain the following high-probability guarantees in terms of the spectral residual of $A$.
\begin{theorem}[High-probability upper bound on $\xone$ with adaptive sketching]
\label{theoremboundfirstorderrandom}
For a Gaussian embedding, under the hypotheses of Lemma~\ref{lemmaupperboundZf} and provided that $\lambda \gre 2 \mu c_g^2 R_k^2(A)$, it holds with probability at least $1-6e^{-k}$ that
\begin{align}
\label{eqnboundfirstorderrandomgaussian}
    \frac{{\|\xone \!-\! x^*\|}_2}{{\|x^*\|}_2} \less \sqrt{\frac{c_g^2 \mu}{2\lambda}} \cdot R_k(A) \cdot \min\!\left\{1, \frac{{\|\xzero\!-\!x^*\|}_2}{\|x^*\|_2}\right\}\,.
\end{align}
For a SRHT embedding, under the hypotheses of Lemma~\ref{lemmaresidualsrht} and provided that $\lambda \gre 2 \mu c_s^2 R_k^2(A)$, it holds with probability at least $1-\frac{5}{n}$ that 
\begin{align}
\label{eqnboundfirstorderrandomsrht}
    \frac{{\|\xone - x^*\|}_2}{{\|x^*\|}_2} \less \sqrt{\frac{c_s^2 \mu}{2\lambda}} \cdot R_k(A)\cdot \min\!\left\{1, \frac{{\|\xzero - x^*\|}_2}{\|x^*\|_2}\right\}\,.
\end{align}
\end{theorem}
The regime $\lambda/\mu \ll R_k^2(A)$ corresponds to a sketch size which is too small relatively to the regularization parameter $\lambda$ and the spectral decay of $A$. In this regime, it can be shown that $\frac{\|\xone - x^*\|_2}{\|x^*\|_2} \lesssim \frac{\mu}{\lambda} \cdot R_k(A)$, and this is a weaker bound for small values of $\frac{\lambda}{\mu}$.

\subsection{Extension to kernel methods}

The first-order optimality conditions $x^* = - \lambda^{-1} A^\top \nabla f(Ax^*)$ yield that $x^* \in \text{range}(A^\top)$. Therefore, the primal program~\eqref{eqprimal} can be solved through the kernel formulation
\begin{align}
\label{eqprimalkernel}
    w^* \in \argmin_{w \in \mathbb{R}^n} \left\{f(K w) + \frac{\lambda}{2} w^\top K w \right\}\,,
\end{align}
where $K=AA^\top$, and then setting $x^* = A^\top w^*$. Given an embedding $\wtilde S \in \real^{n \times m}$, we consider the sketched version of~\eqref{eqprimalkernel},
\begin{align}
\label{eqsketchedprimalkernel}
    \alpha^*_K \in \argmin_{\alpha \in \mathbb{R}^m} \left\{f(K \wtilde{S} \alpha) + \frac{\lambda}{2} \alpha^\top \wtilde{S}^\top K \wtilde{S} \alpha \right\}\,.
\end{align}
The sketched program~\eqref{eqsketchedprimalkernel} is equivalent to~\eqref{eqsketchedprimal} with the adaptive embedding $S=A^\top \wtilde S$, i.e., $\alpha^* = \alpha^*_K$ and $\xone = A^\top \wone$ where $\wone \defn -\lambda^{-1} \nabla f(K\wtilde S \alpha^*_K)$. In other words, we have equivalence between adaptive sketching of the data matrix $A$ and oblivious sketching of the Gram matrix $K$. This equivalence naturally extends to any kernel method with a smooth loss function. We recall the concepts necessary to the exposition of our results, and we refer the reader to the books~\cite{kernelprobastats, gu2013smoothing, wahba1990spline} for more details and background. Given a measurable space $\Omega$ endowed with a probability distribution $\mathbb{P}$, we consider the space $L^2(\Omega, \mathbb{P})$ of real-valued functions over $\Omega$ which are square-integrable with respect to $\mathbb{P}$, and we let $\mathcal{H} \subset L^2(\Omega, \mathcal{P})$ be a reproducing kernel Hilbert space (RKHS) with reproducing kernel $\mathcal{K} : \Omega \times \Omega \to \real$ and associated norm $\|\cdot\|_\mathcal{H}$. Given $\omega_1, \dots, \omega_n \in \Omega$, we aim to solve the (infinite-dimensional) kernel program $h^* \defn \argmin_{h \in \mathcal{H}} f(\{h(\omega_i)\}_{i=1}^n) + \frac{\lambda}{2} \|h\|_{\mathcal{H}}^2$. This kernel program occurs in many widely used machine learning contexts (e.g., kernel ridge regression, kernel support vector machines with smooth hinge loss or kernel logistic regression). The representer theorem~\cite{kimeldorf1971some} states that $h^*$ belongs to the span of the functions $\mathcal{K}(\cdot, \omega_1), \dots, \mathcal{K}(\cdot, \omega_n)$, i.e., there exists $w^* \in \real^n$ such that $h^* = \sum_{i=1}^n \mathcal{K}(\cdot,\omega_i) w^*_i$, and one can solve instead the finite-dimensional program~\eqref{eqprimalkernel} with the empirical kernel matrix $K \defn \{\mathcal{K}(\omega_i,\omega_j)\}_{i,j}$. Based on a low-dimensional solution $\alpha_K^*$ to~\eqref{eqsketchedprimalkernel}, we define the zero- and first-order estimators of $h^*$ as $\what h^{(0)} \defn \sum_{i=1}^n \mathcal{K}(\cdot,\omega_i)\wzero_i$ where $\wzero \defn \wtilde S \alpha^*_K$, and, $\what h^{(1)} \defn \sum_{i=1}^n \mathcal{K}(\cdot, \omega_i) \wone_i$ where $\wone \defn -\lambda^{-1} \nabla f(K\wtilde S \alpha^*_K)$. Let $K_h$ be a square-root of $K$, i.e., $K = K_h K_h^\top$, and we introduce its corresponding restricted singular value
\begin{align}
    \mathcal{Z}_{f,K} \defn \sup_{\Delta \in \mathcal{C}_{z_K^*}} \|P_{K_h^\top \wtilde S}^\perp K_h^\top \Delta\|_2\,.
\end{align}
where $z_K^* \defn \argmin_{z}\left\{f^*(z) + \frac{1}{2\lambda} z^\top K z\right\}$. We recall that $\|\sum_{i=1}^n \mathcal{K}(\cdot, \omega_i) w_i\|_\mathcal{H}=\sqrt{w^\top K w}$ for any $w \in \real^n$. We obtain the following recovery guarantee as a function of the spectral decay of $\mathcal{Z}_{f,K}$.
\begin{theorem}
\label{thmainkernel}
Let $\wtilde S \in \real^{n \times m}$ be an embedding matrix, and let $\alpha^*_K$ be a minimizer of the sketched kernel program~\eqref{eqsketchedprimalkernel}. Under the condition $\lambda \gre 2 \mu \mathcal{Z}_{f,K}^2$, it holds that
\begin{align}
    \frac{\|\what h^{(1)} - h^*\|_\mathcal{H}}{\|h^*\|_\mathcal{H}} \less \sqrt{\frac{\mu}{2\lambda}} \cdot \mathcal{Z}_{f,K} \cdot \min\!\left\{1, \frac{\|\what h^{(0)} - h^*\|_\mathcal{H}}{\|h^*\|_\mathcal{H}}\right\}\,.
\end{align}
\end{theorem}
\begin{proof}
Define $x^* \defn \argmin_x f(K_h x) + \frac{\lambda}{2} \|x\|_2^2$, and $S\defn K_h^\top \wtilde S$. Note that $\alpha^*_K \in \argmin_{\alpha\in \real^m} f(K_h S \alpha) + \frac{\lambda}{2} \|S\alpha\|_2^2$. Set $\xone = - \lambda^{-1} K_h^\top \nabla f(K_h S\alpha^*_K)$ and $\xzero = S \alpha^*_K$. Then, using the results of Theorem~\ref{theoremzeroversusfirstorder} with $A = K_h$, we obtain $\frac{{\|\xone-x^*\|}_2}{{\|x^*\|}_2} \less \sqrt{\frac{\mu}{2 \lambda}} \mathcal{Z}_{f,K} \min\!\left\{1, \frac{\|\xzero - x^*\|_2}{\|x^*\|_2}\right\}$, provided that $\lambda \gre 2\mu \mathcal{Z}_{f,K}^2$. We conclude by using the identities $\|\what h^{(1)} - h^*\|_\mathcal{H} = \|\xone - x^*\|_2$, $\|\what h^{(0)} - h^*\|_\mathcal{H} = \|\xzero - x^*\|_2$  and $\|h^*\|_\mathcal{H} = \|x^*\|_2$.
\end{proof}
Similarly to Theorem~\ref{theoremboundfirstorderrandom}, the results of Theorem~\ref{thmainkernel} along with the concentration bounds in Lemmas~\ref{lemmaupperboundZf} and~\ref{lemmaresidualsrht} yield high-probability bounds on the recovery error $\frac{\|\what h^{(1)} - h^*\|_\mathcal{H}}{\|h^*\|_\mathcal{H}}$ in terms of the spectral decay of the kernel class $\mathcal{K}$. Typical decays encountered in practice are polynomial (e.g., Sobolev kernel) and exponential (e.g., Gaussian kernel).

Oblivious sketching of kernel matrices with Gaussian embeddings or the SRHT has already been considered in~\cite{yang2017randomized}, in the context of kernel ridge regression: the authors analyze the statistical performance of the zero-order estimator $\what h^{(0)}$ as measured by the in-sample predictive norm $\sqrt{\sum_{i=1}^n (\what h(\omega_i) - h^*(\omega_i))^2}$. We contribute to this set of results by showing that the first-order estimator $\what h^{(1)}$ has better recovery error than $\what h^{(0)}$ in the RKHS norm. We leave as an open problem a more extensive comparison of $\what h^{(1)}$ and $\what h^{(0)}$ in the predictive norm.

\section{Smooth, convex optimization in oblivious random subspaces}
\label{sectionoblivious}

\subsection{Limited performance of $\xzero$ and $\xone$ with oblivious sketching}

We first provide an upper bound on the restricted singular value $\mathcal{Z}_f$ for an oblivious Gaussian embedding, which is significantly weaker than in the adaptive case.
\begin{lemma}[Restricted singular value with oblivious embeddings]
\label{lemmaupperboundZfoblivious}
Let $S \in \real^{d\times m}$ be a matrix with i.i.d.~Gaussian entries. Then, it holds with probability at least $1-2e^{-(d-m)}$ that 
\begin{align}
\label{eqnupperboundZfoblivious}
    \mathcal{Z}_f \less c \cdot \left(\frac{\omega(A^\top \mathcal{C}_{z^*})}{\sqrt{d}} + \sqrt{\frac{d-m}{d}} \cdot \|A\|_2\right)\,,
\end{align}
for some universal constant $c > 0$.  
\end{lemma}
Even for a small width $\omega(A^\top \mathcal{C}_{z^*})$, the upper bound~\eqref{eqnupperboundZfoblivious} scales at least as $\sqrt{\frac{d-m}{d}}\cdot \|A\|_2$, and this is large unless $d \approx m$, which defeats the purpose of sketching. This suggests a limited performance of the estimators $\xzero$ and $\xone$ with oblivious embeddings. We formalize this statement by deriving lower bounds on their respective recovery error. For conciseness, we focus on the expected relative error, where the expectation is taken with respect to the randomness of the embedding matrix $S$.

\begin{theorem}[Lower bound on the recovery error of $\xzero$ with oblivious sketching]
\label{theoremboundzeroorderrandom}
It holds for both Gaussian embeddings and the SRHT that
\begin{align}
\label{eqnlowerboundzeroorder}
    \mathbb{E}_S\!\left\{\frac{\|\xzero - x^*\|_2^2}{\|x^*\|_2^2}\right\} \gre 1-\frac{m}{d}\,.
\end{align}
\end{theorem}
\begin{proof}
Note that $\xzero = P_S \xzero$ and thus, $\|\xzero - x^*\|_2^2 = \|P_S (\xzero - x^*)\|_2^2 + \|P_S^\perp x^*\|_2^2$, i.e., $\|\xzero - x^*\|^2_2 \gre \|P_S^\perp x^*\|^2_2$. Using that $\mathbb{E}_S\|P_S^\perp x^*\|_2^2 = (1-\frac{m}{d})\|x^*\|_2^2$ for both embeddings yields the claim.
\end{proof}
The estimator $\xzero$ lies in a low-dimensional oblivious random subspace and does not recover the residual projection of $x^*$ onto $P_S^\perp$: the error $\frac{\|P_S^\perp x^*\|^2_2}{\|x^*\|^2_2} \approx 1-\frac{m}{d}$ is large unless $m \approx d$. We provide next a lower bound on the \emph{worst-case} recovery error of $\xone$, for which we assume the function $f$ to be $\gamma$-strongly convex, i.e., $f(y) \gre f(x) + \langle \nabla f(x), y-x\rangle + \frac{\gamma}{2}\|y-x\|_2^2$ for any $x,y \in \real^d$. Although a similar result could hold for the SRHT, the proof would involve different technical arguments and we specialize our result to Gaussian embeddings for conciseness.

\begin{theorem}[Lower bound on the recovery error of $\xone$ with oblivious sketching]
\label{theoremlowerboundfirstorderobliviousrandom}
Let $S \in \real^{d \times m}$ be a matrix with i.i.d.~Gaussian entries. Then, it holds that 
\begin{align}
\label{eqnlowerboundfirstorder}
    \sup_{f \in \mathcal{F}_{\gamma,\mu}} \, \mathbb{E}_S\!\left\{\frac{\|\xone - x^*\|_2^2}{\|x^*\|_2^2}\right\} \gre \left(1-\frac{m}{d}\right)^3 \cdot \frac{\sigma_1^4}{(\sigma_1^2 + \frac{2\lambda}{\gamma})^2}\,,
\end{align}
where $\mathcal{F}_{\gamma,\mu}$ is the set of real-valued functions defined over $\real^n$, which are $\gamma$-strongly convex and $\mu$-smooth.
\end{theorem}
As for the zero-order estimator, the first-order estimator $\xone$ with oblivious embeddings has a (worst-case) recovery error bounded away from $0$, unless $m \approx d$. In both lower bounds~\eqref{eqnlowerboundzeroorder} and~\eqref{eqnlowerboundfirstorder}, the limiting factor $1-\frac{m}{d}$ is primarily due to the bias $\mathbb{E}\{P_S^\perp\} = \left(1-\frac{m}{d}\right) I_d \neq 0$, which we aim to address next.

\subsection{Improved performance through unbiased oblivious embeddings}
\label{sectionunbiasedoblivioussketching}

We consider a variant of the sketched primal program~\eqref{eqsketchedprimal}, given by
\begin{align}
\label{eqnrescaledsketchedprimal}
    \alpha^*_\dagger \defn \argmin_{\alpha_\dagger \in \real^m}\left\{ f(AQ\alpha_\dagger) + \frac{\lambda}{2}\|\alpha_\dagger\|_2^2 \right\}\,,
\end{align}
where $Q \in \real^{d \times m}$ is a sketching matrix. The low-dimensional formulations~\eqref{eqsketchedprimal} and~\eqref{eqnrescaledsketchedprimal} only differ in the choice of the regularization term, that is, $\frac{\lambda}{2} \|S\alpha\|_2^2$ versus $\frac{\lambda}{2} \|\alpha_\dagger\|_2^2$. We define the corresponding zero- and first-order estimators based on the (unique) low-dimensional solution $\alpha_\dagger^*$ as
\begin{align}
    \xzero_\dagger \defn Q \alpha_\dagger^*\,,\qquad \xone_\dagger \defn -\frac{1}{\lambda} A^\top \nabla f(AQ\alpha^*_\dagger)\,.
\end{align}
The next result relates the two low-dimensional programs~\eqref{eqsketchedprimal} and~\eqref{eqnrescaledsketchedprimal} more precisely. For an embedding matrix $S$, we denote by $U_S \Sigma_S V_S^\top$ a thin SVD of $S$, and we define its \emph{whitened} version as
\begin{align}
\label{eqnSwhite}
    Q_S \defn U_S V_S^\top\,, 
\end{align}
\begin{lemma}
\label{lemmaequivalencelowdim}
Let $S \in \real^{d \times m}$ be an embedding matrix, and let $\alpha_\dagger^*$ be the low-dimensional solution of~\eqref{eqnrescaledsketchedprimal} with the whitened matrix $Q_S$. Then, it holds that the vector $\alpha^* \defn V_S \Sigma_S^{-1} V_S^\top \alpha_\dagger^*$ is a solution of~\eqref{eqsketchedprimal} with the embedding matrix $S$. Furthermore, the corresponding zero- and first-order estimators of $\alpha^*$ and $\alpha_\dagger^*$ are respectively equal, i.e., 
\begin{align}
    \xzero = \xzero_\dagger\,,\qquad \xone = \xone_\dagger\,.
\end{align}
\end{lemma}
\begin{proof}
We have by first-order optimality conditions that $Q_S^\top A^\top \nabla f(AQ_S\alpha_\dagger^*) + \lambda \alpha_\dagger^* = 0$. Multiplying by $V_S \Sigma_S V_S^\top$ and plugging-in the definition of $\alpha^*$, we obtain $S^\top A^\top \nabla f(AS \alpha^*) + \lambda S^\top S \alpha^* = 0$, i.e., $\alpha^*$ is a solution of~\eqref{eqsketchedprimal}. On the other hand, we have $Q_S\alpha_\dagger^* = U_S V_S^\top V_S \Sigma_S V_S^\top \alpha^* = S\alpha^*$, which further implies that $\xzero = \xzero_\dagger$ and $\xone=\xone_\dagger$.
\end{proof}
The low-dimensional formulation~\eqref{eqsketchedprimal} is equivalent to~\eqref{eqnrescaledsketchedprimal} with the whitened matrix $Q_S$, in the sense that they yield the same zero- and first-order estimators. In addition to the whitened matrices $Q_S$ which are in general biased (e.g., $\mathbb{E}\{Q_SQ_S^\top\} = \frac{m}{d}\cdot I_d$ for an oblivious Gaussian embedding $S$), the formulation~\eqref{eqnrescaledsketchedprimal} can incorporate any sketching matrix $Q$, and in particular, random matrices such that $\mathbb{E}\{QQ^\top\} = I_d$. This yields a larger class of zero- and first-order estimators, which may overcome the aforementioned limited performance of $\xzero$ and $\xone$ with oblivious embeddings. The estimator $\xone_\dagger$ with an unbiased oblivious Gaussian embedding has already been considered in~\cite{zhang2013recoveringCOLT}, and we recall their main result (Theorem~6 in~\cite{zhang2013recoveringCOLT}).
\begin{theorem}[Recovery error of the unbiased oblivious estimator $\xone_\dagger$]
\label{theoremzhangrecovering}
Suppose that the function $f$ is separable and that the solution $x^*$ lies in the span of the top $k$ right singular vectors of $A$. Let $Q \in \real^{d \times m}$ be a matrix with i.i.d.~Gaussian entries $\mathcal{N}(0,1/m)$. Then, provided that $m \gre 32 d_{\lambda/\mu} \log(2d/\delta)$, it holds that
\begin{align}
\label{equpperboundxonedagger}
    \frac{\|\xone_\dagger - x^*\|_2}{\|x^*\|_2} \less \sqrt{128 \log(2d/\delta)} \cdot \sqrt{\frac{d_{\lambda/\mu}}{m}} \cdot \left(1+\sqrt{\frac{\lambda}{\mu \sigma_k^2}}\right)\,,
\end{align}
with probability at least $1-\delta$.
\end{theorem}
The best-case upper bound~\eqref{equpperboundxonedagger} scales as $\sqrt{\frac{d_{\lambda/\mu}}{m}}$. That is, the sketch size $m$ needs to scale at least as the effective dimension which can be much smaller than the ambient dimension $d$. This significantly improves on the guarantees for $\xzero$ and $\xone$ with oblivious embeddings. Furthermore, the oblivious zero-order estimator $\xzero_\dagger$ has worse performance than $\xone_\dagger$ (see Theorem~3 in~\cite{zhang2013recoveringCOLT}). The assumptions of Theorem~6 in~\cite{zhang2013recoveringCOLT} have been slightly relaxed by the same authors (see Theorem~2 and Corollary 3 in~\cite{zhang2014recoveringIEEE}): provided that the projection of $x^*$ onto the subspace orthogonal to the top $k$ right singular vectors of $A$ has small enough norm, then the recovery error scales as $\sqrt{\frac{k}{m}}$.

Let us now compare the guarantees for the oblivious estimator $\xone_\dagger$ and for the adaptive estimator $\xone$. Besides being independent of any assumption on $x^*$, the upper bound~\eqref{eqnboundfirstorderrandomgaussian} scales as $\sqrt{\frac{\mu}{\lambda}} \cdot R_{m/2}(A)$. In light of the comparison between the characteristic quantities $\sqrt{\frac{d_{\lambda/\mu}}{m}}$ and $\sqrt{\frac{\mu}{\lambda}} \cdot R_{m/2}(A)$ provided in Table~\ref{tablecomparisonspectralresidualeffdim}, the guarantees for the adaptive estimator $\xone$ are in general stronger than the best-case guarantee~\eqref{equpperboundxonedagger} for the unbiased oblivious estimator $\xone_\dagger$.  

In light of the results of Theorems~\ref{theoremzeroversusfirstorder},~\ref{theoremboundfirstorderrandom},~\ref{theoremboundzeroorderrandom},~\ref{theoremlowerboundfirstorderobliviousrandom} and~\ref{theoremzhangrecovering}, we obtain that among the different possible estimators $\xzero$ and $\xone$ with oblivious and adaptive sketching, and, $\xzero_\dagger$ and $\xone_\dagger$ with unbiased oblivious sketching, the strongest guarantees are obtained for the first-order estimator $\xone$ with an adaptive embedding. We compare numerically the performance of the adaptive estimator $\xone$ and the unbiased oblivious estimator $\xone_\dagger$ with Gaussian embeddings and for polynomial and exponential decays. We use $n=1000$, $d=2000$ and we generate data matrices with respective spectral decay $\sigma_j = \sqrt{n} e^{-0.05 j}$ and $\sigma_j = \sqrt{n} j^{-1}$. We consider two convex smooth loss functions: the logistic function $f(w) = n^{-1} \sum_{i=1}^n \log(1+e^{-y_i w_i})$ where $y \in \{\pm 1\}^n$, and, a second loss function $f(w) = (2n)^{-1} \sum_{i=1}^n (w_i)^2_+ - 2w_i y_i$ (where $a_+ \defn \max\{a,0\}$ for $a\in \real$) which can be seen as the convex relaxation of the penalty $\frac{1}{2}\|w_+ - y\|_2^2$ for fitting a shallow neural network with a ReLU non-linearity~\cite{ergen2019convex,ergen2020convex,pilanci2020neural}. We report results in Figure~\ref{figsyntheticresults}. The adaptive estimator $\xone$ has better empirical performance than the oblivious one $\xone_\dagger$.

\begin{figure}[h!]
\centering
\includegraphics[width=0.95\columnwidth]{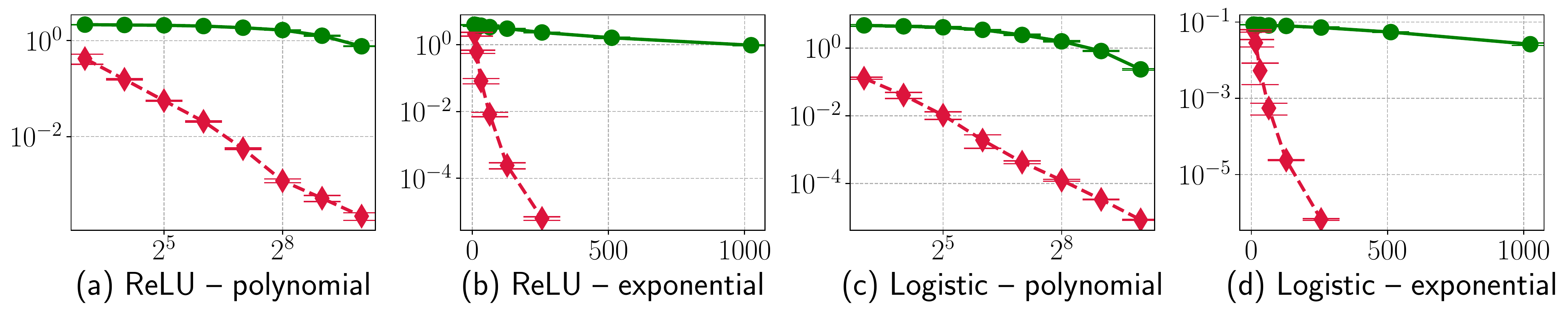}
\caption{Relative recovery error versus sketching dimension $m \in \{2^k \mid 3\less k\less10\}$ of $\xone$ (red diamonds) and $\xone_\dagger$ (green circles), for the ReLU-type and logistic loss functions, and the exponential and polynomial decays. We use $\lambda=10^{-4}$ for all simulations. Results are averaged over $10$ trials. Bar plots show (twice) the empirical standard deviations.}
\label{figsyntheticresults}
\end{figure}

\section{Algorithms for smooth convex objectives in adaptive random subspaces}
\label{sectionalgorithms}

\subsection{Prototype algorithm for adaptive sketching}

A standard quantity to characterize whether a convex program can be solved efficiently is its condition number~\cite{boyd2004convex}, which, for the primal~\eqref{eqprimal} and sketched program~\eqref{eqsketchedprimal}, is respectively given by
\begin{align}
    \kappa \defn \sup_{x \in \real^d}\,\frac{\lambda + \sigma_{1\!}\left( A^\top \nabla^2 f(Ax) A\right)}{\lambda + \sigma_{d}\!\left( A^\top \nabla^2 f(Ax) A\right)}\,, \qquad \kappa_S \defn \sup_{\alpha \in \real^m}\,\frac{\sigma_{1}\!\left( S^\top (\lambda I_d + A^\top \nabla^2 f(A S \alpha) A) S\right)}{\sigma_{m}\!\left( S^\top ( \lambda I_d + A^\top \nabla^2 f(A S \alpha) A) S\right)}\,.
\end{align}
The latter can be significantly larger than $\kappa$, up to $\kappa_S \approx \kappa \cdot\frac{\sigma_{1}\left(S^\top S\right)}{\sigma_{m}\left(S^\top S\right)} \gg \kappa$. According to Lemma~\ref{lemmaequivalencelowdim}, we can solve instead the optimization problem~\eqref{eqnrescaledsketchedprimal} with the whitened matrix $Q_S$, and we have $\xone = -\lambda^{-1} A^\top \nabla f(A Q_S\alpha_\dagger^*)$. Fortunately, the re-scaled sketched program~\eqref{eqnrescaledsketchedprimal} with $Q=Q_S$ is numerically well-conditioned: in fact, it is even better conditioned than the original primal program~\eqref{eqprimal}.
\begin{proposition}
\label{prop:numericalconditioning}
The condition number $\kappa_\dagger$ of the re-scaled sketched program~\eqref{eqnrescaledsketchedprimal}
\begin{align}
    \kappa_\dagger \defn \sup_{\alpha \in \real^m} \frac{\lambda + \sigma_1\!\left(Q_S^\top A^\top \nabla^2 f(AQ_S\alpha)AQ_S\right)}{\lambda + \sigma_m\!\left(Q_S^\top A^\top \nabla^2 f(AQ_S\alpha)AQ_S\right)}
\end{align}
satisfies $\kappa_\dagger \less \kappa$ almost surely.
\end{proposition}
\begin{proof}
Fix $\alpha \in \real^m$. Using $\|Q_S\|_2 \less 1$, we obtain 
\begin{align*}
    &\sigma_1(Q_S^\top A^\top \nabla^2 f(AQ_S\alpha) A Q_S) \less \sigma_1(A^\top \nabla^2 f(AQ_S\alpha) A)\,,\\  
    &\sigma_m(Q_S^\top A^\top \nabla^2 f(AQ_S \alpha) AQ_S) \gre \sigma_d(A^\top \nabla^2 f(AQ_S\alpha)A)\,.
\end{align*}
Consequently, 
\begin{align*}
    \frac{\lambda + \sigma_1(Q_S^\top A^\top \nabla^2 f(AQ_S\alpha) A Q_S)}{\lambda + \sigma_m(Q_S^\top A^\top \nabla^2 f(AQ_S \alpha) AQ_S)} \less \frac{\lambda + \sigma_1(A^\top \nabla^2 f(AQ_S\alpha) A)}{\lambda + \sigma_d(A^\top \nabla^2 f(AQ_S\alpha)A)}\,. 
\end{align*}
Taking the supremum over $\alpha \in \real^m$ in both sides of the latter inequality, we obtain $\kappa_\dagger \less \sup_{\alpha \in \real^m}\,\frac{\lambda + \sigma_1(A^\top \nabla^2 f(AQ_S\alpha) A)}{\lambda + \sigma_d(A^\top \nabla^2 f(AQ_S\alpha)A)}$. We conclude using the fact that the latter right-hand side is smaller than $\kappa$.
\end{proof}
\begin{algorithm}
\label{alg:MainAlgorithm}
\DontPrintSemicolon
\SetKwInOut{Input}{Input}
\Input{Data matrix $A \in \mathbb{R}^{n \times d}$, random matrix $\wtilde{S} \in \mathbb{R}^{n \times m}$ and regularization parameter $\lambda > 0$.}
Compute the sketching matrix $S = A^\top \wtilde{S}$.\\
Compute a thin SVD $S=U_S \Sigma_S V_S^\top$ and set $Q_S = U_S V_S^\top$.\\
Solve the convex optimization problem~\eqref{eqnrescaledsketchedprimal} with $Q=Q_S$, and return $\xone = -\frac{1}{\lambda} A^\top \nabla f\left(A Q_S\alpha^*_\dagger \right)$.
\caption{Prototype algorithm for adaptive sketching in the smooth case.}
\end{algorithm}
Algorithm~\ref{alg:MainAlgorithm} is decomposed into three steps: forming the sketch $A Q_S$, solving the low-dimensional program~\eqref{eqnrescaledsketchedprimal}, and, mapping $\alpha_\dagger^*$ to $\xone$. The last step is, in general, relatively cheap, as it only requires a matrix-vector multiplication with $A^\top$ and a gradient call to $f$. In total, the algorithm requires three passes over the entire data matrix $A$. Depending on the choice and structure of the random embedding, the sketching part has different computational costs, as discussed next. We denote by $\text{nnz}(A)$ the number of non-zero entries of $A$.

\subsubsection{Computational complexity for Gaussian embeddings} 
Forming $S=A^\top \wtilde S$ takes time $\mathcal{O}(m \cdot \text{nnz}(A))$. The cost of computing the SVD of $S$ is $\mathcal{O}(d \cdot m^2)$. The matrix multiplication $A\cdot Q_S$ takes time $\mathcal{O}(m \cdot \text{nnz}(A))$. Therefore, the total complexity is given by 
\begin{align}
    \mathcal{O}(2 \cdot m \cdot \text{nnz}(A) + d \cdot m^2)\,.
\end{align}
For a dense matrix $A$, this results in $\mathcal{O}(2mnd + d m^2) = \mathcal{O}(2mnd)$ floating point operations, and the cost is dominated by the sketching part. For a sparse enough matrix $A$ with $\text{nnz}(A) \lesssim dm$, the total cost is $\mathcal{O}(d m^2)$.

\subsubsection{Computational complexity for the SRHT} Differently from dense and unstructured embeddings, forming $S=A^\top \wtilde S$ takes time $\mathcal{O}(\log m \cdot \text{nnz}(A))$. The total time complexity is then given by
\begin{align}
    \mathcal{O}(\log m \cdot \text{nnz}(A) + m \cdot \text{nnz}(A) + d \cdot m^2) = \mathcal{O}(m \cdot \text{nnz}(A) + d \cdot m^2)\,,
\end{align}
which is always smaller than the cost of Gaussian embeddings. For a dense matrix $A$, this results in $\mathcal{O}(mnd)$ floating point operations, and this is half the cost with Gaussian embeddings. For a sparse enough matrix $A$ with $\text{nnz}(A) \lesssim dm$, the cost similarly scales as $\mathcal{O}(d m^2)$.

\subsection{Improved algorithms}

\subsubsection{Iterative method and almost exact recovery of the optimal solution} The estimator $\xone$ satisfies a guarantee of the form $\|\xone-x^*\|_2 \less \varepsilon\, \|x^*\|_2$ with high probability, and with $\varepsilon < 1$ provided that $m$ is large enough relatively to $\frac{\lambda}{\mu}$ and the spectral decay of $A$. Here, we extend Algorithm~\ref{alg:MainAlgorithm} to an iterative version which takes advantage of this error contraction, and which does not incur additional memory requirements, at the expense of additional time complexity.

\begin{algorithm}
\label{alg:Algorithm}
\DontPrintSemicolon
\SetKwInOut{Input}{Input}
\Input{Data matrix $A \in \mathbb{R}^{n \times d}$, random matrix $\wtilde{S} \in \mathbb{R}^{n \times m}$, iterations number $T$, regularization parameter $\lambda > 0$.}
Compute the matrices $Q_S \in \real^{d \times m}$ and $A Q_S$ as in Algorithm~\ref{alg:MainAlgorithm}. Set $\xone_0 = 0$.\\
\For{$t= 1, 2,\dots,T$}{ 
Solve the low-dimensional convex optimization problem
\begin{equation}
\label{eqnintermediatesketchedprimal}
    \alpha^*_{\dagger,t} \defn \argmin_{\alpha_\dagger \in \mathbb{R}^m} \left\{ f(A Q_S \alpha_\dagger + A \xone_{t-1}) + \frac{\lambda}{2} \|\alpha_\dagger + Q_S^\top \xone_{t-1}\|_2^2 \right\} \,.
\end{equation}
Update the solution $\xone_t = -\frac{1}{\lambda} A^\top \nabla f(A Q_S \alpha^*_{\dagger,t} + A \xone_{t-1})$.
}
Return the last iterate $\xone_T$.
\caption{Prototype iterative method for adaptive sketching in the smooth case}
\end{algorithm}

A key advantage of Algorithm~\ref{alg:Algorithm} is that, at each iteration, the same sketching matrix $S$ is used, i.e., the matrices $Q_S$ and $AQ_S$ need to be computed only once, at the beginning of the procedure. The output $\xone_T$ satisfies the following recovery property, whose empirical benefits are illustrated in Figure~\ref{figitpm}.
\begin{theorem}
\label{theoremiterativemethod}
After $T$ iterations of Algorithm~\ref{alg:Algorithm}, provided that $\lambda \gre 2 \mu \mathcal{Z}_f^2$, it holds that
\begin{equation}
\label{eqnbounditerative}
    \frac{{\|\xone_T - x^*\|}_2}{{\|x^*\|}_2} \less \left(\frac{\mu \mathcal{Z}_f^2}{2\lambda}\right)^{\frac{T}{2}}\,.
\end{equation}
For an adaptive Gaussian embedding, under the hypotheses of Lemma~\ref{lemmaupperboundZf} and provided that $\lambda \gre 2 \mu c_g^2 R_k^2(A)$, it holds with probability at least $1-6e^{-k}$ that
\begin{align}
\label{eqnbounditerativegaussian}
    \frac{{\|\xone_T - x^*\|}_2}{{\|x^*\|}_2} \less \left(\frac{c_g^2 \mu R_k^2(A)}{2\lambda}\right)^{\frac{T}{2}}\,.
\end{align}
For an adaptive SRHT, under the hypotheses of Lemma~\ref{lemmaresidualsrht} and provided that $\lambda \gre 2 \mu c_s^2 R_k^2(A)$, it holds with probability at least $1-\frac{5}{n}$ that 
\begin{align}
\label{eqnbounditerativesrht}
    \frac{{\|\xone_T - x^*\|}_2}{{\|x^*\|}_2} \less \left(\frac{c_s^2 \mu R_k^2(A)}{2\lambda}\right)^{\frac{T}{2}}\,.
\end{align}
\end{theorem}
Let us compare the \emph{oracle} complexities of Algorithm~\ref{alg:Algorithm} and first-order methods applied to the high-dimensional program~\eqref{eqprimal} under the assumption that $\frac{\lambda}{\mu \sigma_1^2(A)} \ll 1$ is small, so that the high-dimensional objective function $F$ in~\eqref{eqprimal} is ill-conditioned. After a number $T$ of matrix-vector multiplications of the form $A^\top \nabla f(\cdot)$ (which is equivalent to a gradient call to $F$), Algorithm~\ref{alg:Algorithm} returns an approximate solution $\xone_T$ whose relative error is upper bounded by $\left(\frac{\mu \mathcal{Z}_f^2}{2\lambda}\right)^\frac{T}{2}$. In comparison, given a similar budget of $T$ gradient calls to the objective function $F$, first-order methods applied to the high-dimensional program~\eqref{eqprimal} return an approximate solution $\wtilde x$ whose relative recovery error is upper bounded by $(1-\frac{\lambda}{\mu \sigma_1^2(A)})^\frac{T}{2}$ (see, for instance, Theorem~3.10 in~\cite{bubeck2015convex}). The latter convergence rate is close to $1$, whereas the convergence rate of the sequence $\xone_T$ is bounded away from $1$ provided that the sketch size is large enough.

\subsubsection{Shrinking the spectral residual with the power method} An immediate extension of Algorithms~\ref{alg:MainAlgorithm} and~\ref{alg:Algorithm} consists in using the so-called \emph{power method}~\cite{halko2011}. Given $q \in \mathbb{N}$, the adaptive sketching matrix at power $q$ is defined as $S \defn {(A^\top A)}^q A^\top \wtilde{S}$. The larger $q$, the smaller the approximation error $\|P_S^\perp A^\top\|_2$. More precisely, for a Gaussian embedding $\wtilde S$ with $m=2k$, we have according to Corollary 10.10 in~\cite{halko2011} that
\begin{align}
    \|P_S^\perp A^\top\|_2 \lesssim \sigma_k \left( 1 + \sqrt{\frac{1}{k} \sum_{j > k} \left(\frac{\sigma_j}{\sigma_k}\right)^{2(2q+1)}} \right)^{\frac{1}{2q+1}}
\end{align}
with high probability. The above right-hand side decreases as $q$ increases. Of practical interest are data matrices $A$ of the form $A = \overline{A} + W$, where $\overline{A}$ is a signal with a fast spectral decay, and $W$ is a noise matrix with a relatively small and flat spectral profile. The power method reduces the noise contribution to the spectrum, i.e., it shrinks the error factor $\left( 1 + \sqrt{\frac{1}{k} \sum_{j > k} \left(\frac{\sigma_j}{\sigma_k}\right)^{2(2q+1)}} \right)^{\frac{1}{2q+1}}$. Our algorithms readily incorporate the power method, and we illustrate its empirical benefits in Figure~\ref{figitpm}.
\begin{figure}[h!]
\centering
\includegraphics[width=0.95\columnwidth]{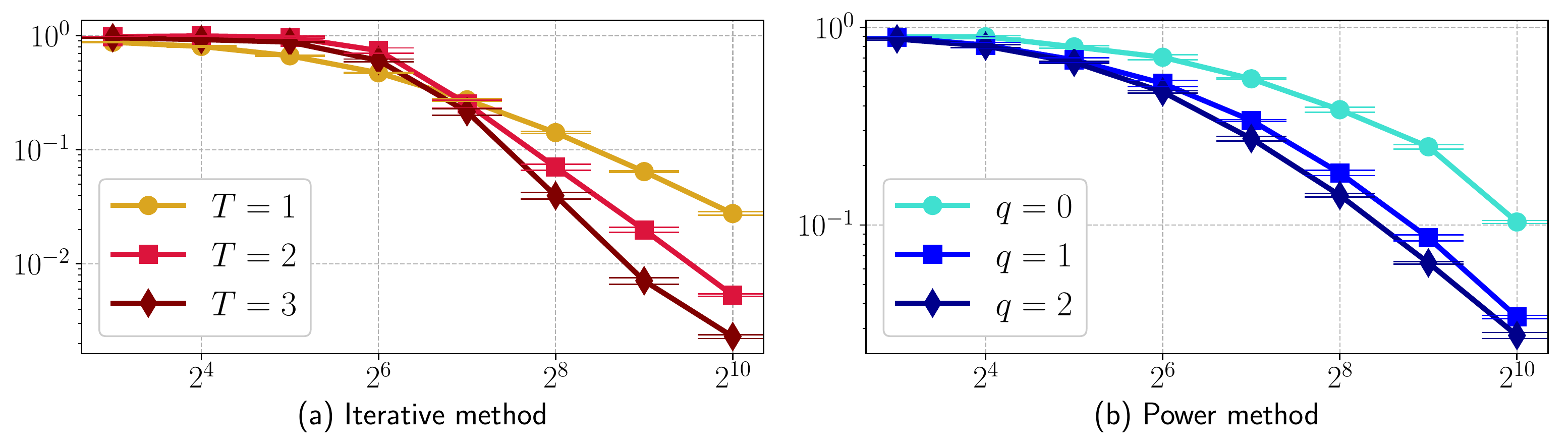}
\caption{Relative recovery error versus sketching dimension $m \in \{2^k \mid 3\less k\less10\}$ of adaptive Gaussian sketching for (a) the iterative method (Algorithm~\ref{alg:Algorithm}) and (b) the power method. We use the MNIST dataset with $50000$ training images and $10000$ testing images, and we map images through $10000$-dimensional through random cosines $\psi$ which approximate the Gaussian kernel~\cite{rahimi2008random}, i.e., $\langle \psi(a), \psi(a^\prime) \rangle \approx \exp(-\gamma \|a-a^2\|_2^2)$, with $\gamma = 0.02$. We perform binary logistic regression for even-vs-odd classification of digits. For the iterative method, we use the sketching matrix $S = (A^\top A)^2 A^\top \wtilde{S}$, where $\wtilde{S}$ is Gaussian i.i.d. That is, we run the iterative method on top of the power method, with $q=2$. We use $\lambda=10^{-5}$. Results are averaged over $20$ trials. Bar plots show (twice) the empirical standard deviations.}
\label{figitpm}
\end{figure}

\section{Non-smooth, convex optimization in adaptive random subspaces}
\label{sectionnonsmooth}

We extend our analysis to the case where the function $f$ is Lipschitz continuous but not necessarily smooth nor even differentiable. In contrast to the smooth case, the dual map $\partial f(AQ_S\alpha^*_\dagger)$ is set-valued and this makes the recovery through dual mapping more challenging. On the other hand, according to Proposition~\ref{propositionfenchelduality}, when the function $f$ is strictly convex, the dual solution $z^*$ lies within the interior of the domain of $f^*$ and consequently, $\mathcal{Z}_f = \|P_S^\perp A^\top\|_2$. As strict convexity holds for most smooth convex objective functions $f$ of practical interest, one may instead expect the tangent cone $\mathcal{T}_{z^*}$ to have a small size in the case of a non-smooth objective function. 

\subsection{Undeterminate estimator through dual mapping}
\begin{theorem}[Deterministic upper bound for non-smooth objectives]
\label{theoremnonsmooth}
Suppose that $f : \real^n \to \real$ is convex and $L$-Lipschitz but not necessarily smooth nor differentiable. Let $S \in \real^{d \times m}$ be an embedding matrix, and $y^*$ be any sketched dual solution, and set $\xone \defn - \lambda^{-1} A^\top y^*$. Then, for any $\lambda > 0$, it holds that
\begin{align}
\label{eqnboundnonsmoothcase}
    \|\xone - x^*\|_2 \less 2 \cdot \frac{L}{\lambda} \cdot \sqrt{\mathcal{Z}_f^2 + \frac{1}{2}\,\mathcal{Z}_f \|P_S^\perp A^\top\|_2} \less \sqrt{6} \cdot \frac{L}{\lambda} \cdot \|P_S^\perp A^\top\|_2\,.
\end{align}
Further, under the additional assumption that $\inf_w f(w) > -\infty$, we have the improved upper bound
\begin{align}
\label{eqnboundnonsmoothcase0in}
    \|\xone - x^*\|_2 \less \sqrt{6} \cdot \frac{L}{\lambda} \cdot \mathcal{Z}_f\,.
\end{align}
\end{theorem}
The upper bound~\eqref{eqnboundnonsmoothcase0in} is, for most cases of interest, weaker than~\eqref{eqnboundfirstorder} as it scales as $\mathcal{O}(\frac{\mathcal{Z}_f}{\lambda})$ in contrast to the bound $\mathcal{O}(\frac{\mathcal{Z}_f}{\sqrt{\lambda}})$ in the smooth case: for small $\lambda$, the former scaling is worse. Furthermore, the upper bound~\eqref{eqnboundnonsmoothcase0in} controls the recovery error $\|\xone - x^*\|_2$ whereas the upper bound in the smooth case controls the relative recovery error $\frac{\|\xone - x^*\|_2}{\|x^*\|_2}$. The latter is contractive and enables an iterative version for (almost) exact recovery of $x^*$, whereas this approach does not readily extend to the non-smooth case. More importantly, when $f$ is smooth and thus differentiable, the mapping between a low-dimensional solution $\alpha^*_\dagger$ and the estimator $\xone = -\lambda^{-1} A^\top \nabla f(AQ_S\alpha^*_\dagger)$ is well-defined. However, in the non-differentiable case, we only know that the sketched dual solution $y^*$ belongs to the set $\partial f(AQ_S \alpha^*_\dagger)$, i.e., $\xone \in -\lambda^{-1} A^\top \partial f(AQ_S\alpha^*_\dagger)$, and one cannot directly compute $\xone$ based on the low-dimensional solution $\alpha^*_\dagger$ and the dual mapping. Furthermore, picking an arbitrary subgradient $g \in \partial f(AQ_S\alpha^*_\dagger)$ yields an estimator $\what x = -\lambda^{-1} A^\top g$ with weaker recovery guarantees. 
\begin{corollary}
\label{corollaryundeterminedestimator}
Suppose that $\inf_w f(w) > -\infty$. Let $\alpha^*_\dagger$ be a minimizer of the sketch primal program~\eqref{eqnrescaledsketchedprimal} with $Q=Q_S$. Pick any estimator $\what x$ in the set $-\lambda^{-1} A^\top \partial f(AQ_S \alpha^*_\dagger)$. It holds that
\begin{align}
    \|\what x - x^*\|_2 \less \sqrt{6} \cdot \frac{L}{\lambda} \cdot \mathcal{Z}_f + \text{diam}\!\left(\lambda^{-1} A^\top \partial f(AQ_S \alpha^*_\dagger)\right)\,.
\end{align}
\end{corollary}
\begin{proof}
Fix $\what x \in -\lambda^{-1} A^\top \partial f(AQ_S \alpha^*_\dagger)$. The first-order estimator $\xone = -\lambda^{-1} A^\top y^*$ also belongs to $-\lambda^{-1} A^\top \partial f(AQ_S \alpha^*_\dagger)$, so that $\|\what x - \xone\|_2 \less \text{diam}\!\left(\lambda^{-1} A^\top \partial f(AQ_S \alpha^*_\dagger)\right)$. According to Theorem~\ref{theoremnonsmooth}, we also have $\|\xone - x^*\|_2 \less \sqrt{6}\cdot \frac{L}{\lambda} \cdot \mathcal{Z}_f$. The claim then follows from the triangular inequality $\|\what x - x^*\|_2 \less \|\what x - \xone\|_2 + \|\xone - x^*\|_2$.
\end{proof}
According to Corollary~\ref{corollaryundeterminedestimator}, an arbitrary estimator $\what x \in -\lambda^{-1} A^\top \partial f(AQ_S \alpha^*_\dagger)$ may perform poorly, even when $\mathcal{Z}_f$ is small, as soon as the diameter of the set $\lambda^{-1} A^\top \partial f(AQ_S \alpha^*_\dagger)$ is relatively large. One may wonder whether this upper bound is tight: although we do not provide a lower bound, we carried out extensive numerical simulations with some commonly used non-smooth loss functions, and we obtained poor performance when we picked the easiest-to-compute subgradient in $\partial f(AQ_S \alpha^*_\dagger)$ (see Figure~\ref{fignonsmooth}).

\subsection{Resolving the indeterminacy of the set-valued dual mapping through the sketched dual program}

Based on an adaptive sketch $AQ_S$ and a low-dimensional solution $\alpha^*_\dagger$, the dual map yields a set of candidate approximate solutions $-\lambda^{-1} A^\top \partial f(AQ_S\alpha^*_\dagger)$. We aim to resolve this indeterminacy by computing the optimal subgradient $y^* \in \partial f(AQ_S\alpha^*_\dagger)$ and thus the first-order estimator $\xone = -\lambda^{-1}A^\top y^*$. 

\subsubsection{Easy-to-compute subgradient set $\partial f(AQ_S\alpha_\dagger^*)$ and restricted sketched dual program} We propose to (i) compute a low-dimensional solution $\alpha_\dagger^*$, (ii) to compute the subgradient set $\partial f(AQ_S \alpha_\dagger^*)$, and finally, (iii) to solve the \emph{restricted} sketched dual program
\begin{align}
\label{eqnsketchedualwithQ}
    y^* \in \argmin_{y \in \partial f(AQ_S\alpha^*_\dagger)} \left\{ f^*(y) + \frac{1}{2\lambda} \|Q_S^\top A^\top y\|_2^2 \right\}\,.
\end{align}
This procedure is especially relevant when the function $f$ is separable and when there are only a few number $k$ of indices $i \in \{1, \dots, n\}$ such that the function $f$ is not partially differentiable at $(AS\alpha_\dagger^*)_i$. In this case, the restricted sketched dual program~\eqref{eqnsketchedualwithQ} only involves $k$ dual variable's coordinates. This is reminiscent of the Stochastic Dual Newton Ascent (SDNA) method~\cite{qu2016sdna} which selects at each iteration a random subset of coordinates of the dual variable $y$. Differently, we use the low-dimensional solution $\alpha_\dagger^*$ and the subgradient set $\partial f(AQ_S\alpha_\dagger^*)$ to determine a subset of coordinates of $y$ to optimize and which yields the exact optimal solution $y^*$. Moreover, our perspective is, again, agnostic to the choice of the optimization algorithm, whereas SDNA is itself an optimization method.

More generally, the restricted sketched dual program is relevant for practical cases where (i), given a low-dimensional solution $\alpha^*_\dagger$, computing the subgradient set $\partial f(AQ_S\alpha^*_\dagger)$ can be done efficiently, and (ii) the subgradient set $\partial f(AQ_S\alpha^*_\dagger)$ is, in some sense, small. We discuss some examples below.

\begin{example}[$L_1$-regression]
\label{examplel1}
Given $b \in \real^n$, we consider the objective function $f(w) = \|w-b\|_1$. The subgradient set of $f$ at some $w \in \real^n$ is the Cartesian product $\prod_{i=1}^n \mathcal{I}_i$, where $\mathcal{I}_i = [-1,1]$ if $w_i=b_i$ and $\mathcal{I}_i = \{\text{sign}(w_i-b_i)\}$ otherwise. The restricted sketched dual program~\eqref{eqnsketchedualwithQ} only involves the variables $y_i$ for the indices $i$ such that $(AS\alpha_\dagger^*)_i = b_i$.
\end{example}
\begin{example}[$L_\infty$-regression]
\label{examplelinfty}
Given some target vector $b \in \real^n$, consider the objective function $f(w) = \|w-b\|_\infty$. The subgradient set of $f$ at some $w \in \real^n$ is the convex hull of the vectors $\text{sign}(w_I - b_I)\cdot e_I$, where $e_1, \dots, e_n$ is the canonical basis of $\real^n$ and for all $I \in \argmax_{i=1,\dots,n} |w_i - b_i|$. The restricted sketched dual program~\eqref{eqnsketchedualwithQ} only involves the variables $y_i$ for the indices $i$ such that $|(AS\alpha_\dagger^*)_i-b_i|=\|AS\alpha_\dagger^*-b\|_\infty$.
\end{example}
\begin{example}[Support vector machines]
\label{examplesvm}
Given $b \in \{\pm 1\}^n$, we consider the hinge loss $f(w) = \sum_{i=1}^n \max\{0, 1-w_i b_i\}$. The subgradient set of $f$ at some $w \in \real^n$ is the Cartesian product $\prod_{i=1}^n \mathcal{I}_i$, where $\mathcal{I}_i = \{-b_i\}$ if $1-w_i b_i > 0$, $\mathcal{I}_i =\{0\}$ if $1-w_i b_i < 0$ and $\mathcal{I}_i = [0,-b_i]$ if $1-w_i b_i=0$. The restricted sketched dual program~\eqref{eqnsketchedualwithQ} only involves the variables $y_i$ for the indices $i$ such that $1=(AS\alpha_\dagger^*)_i b_i$.
\end{example}
For the loss functions of Examples~\ref{examplel1},~\ref{examplelinfty} and~\ref{examplesvm}, we compare the empirical performance of the estimator $\xone = -\lambda^{-1}A^\top y^*$ and an arbitrary estimator $\what x \in -\lambda^{-1}A^\top \partial f(AQ_S\alpha_\dagger^*)$, and we report results in Figure~\ref{fignonsmooth}. We observe that $\xone$ has an increasingly stronger performance compared to $\what x$ as the sketch size increases.

\begin{figure}[h!]
\centering
\includegraphics[width=0.95\columnwidth]{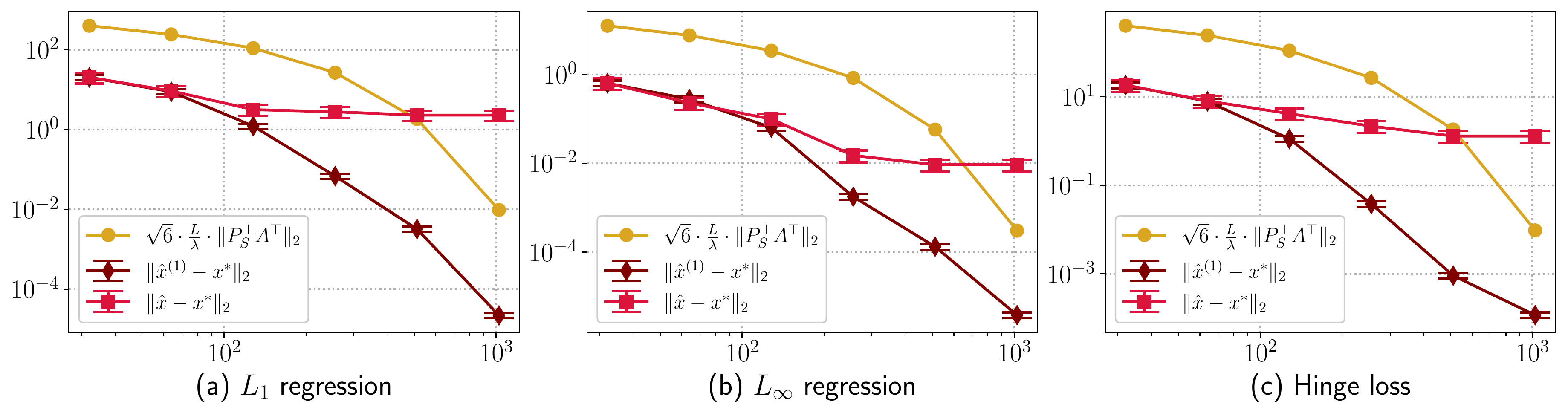}
\caption{Recovery error versus sketching dimension $m \in \{2^k \mid 5\less k\less10\}$ of the adaptive estimator $\xone=-\lambda^{-1}A^\top y^*$ versus an estimator $\what x = -\lambda^{-1} A^\top g$ where $g \in \partial f(AQ_S\alpha_\dagger^*)$ is an arbitrary subgradient (the 'easiest-to-compute'). We use $\lambda=0.01$, $n=1000$ and $d=2000$. We generate a data matrix $A$ with exponential spectral decay $\sigma_j=0.98^j$, and we use an adaptive Gaussian embedding $S=A^\top \wtilde S$. We consider (a) $L_1$-regression (Example~\ref{examplel1}), (b) $L_\infty$-regression (Example~\ref{examplelinfty}) and (c) SVM classification (Example~\ref{examplesvm}). Results are averaged over $20$ trials. Bar plots show (twice) the empirical standard deviations.}
\label{fignonsmooth}
\end{figure}

\subsubsection{Solving the plain sketched dual program} When the entire subgradient set $\partial f(AQ_S\alpha^*_\dagger)$ cannot be efficiently computed, one can instead solve directly the plain sketched dual program without restricting $y$ to lie within the subgradient set $\partial f(AQ_S\alpha_\dagger^*)$ and without even computing a low-dimensional solution $\alpha_\dagger^*$. Of practical interest are, for instance, the functions $f$ which, up to a translation, are the support function of a convex set $\mathcal{C}$, i.e.,
\begin{align}
\label{eqnsupportfunction}
    f(w) = \sup_{z \in \mathcal{C}} \, z^\top (w - b)\,,
\end{align}
for which $\partial f(w) = \argmax_{z \in \mathcal{C}} \, z^\top (w-b)$. This class of function includes in particular distributionally robust objective functions (e.g., conditional Value-at-Risk) and deterministic robust counterparts \cite{ben2009robust,pilanci2009structured,pilanci2010structured}  for which computing the entire set of worst-case distributions is, in general, expensive. Hence, one can alternatively obtain $y^*$ by solving the \emph{unrestricted} constrained quadratic program
\begin{align}
    y^* \in \argmin_{y \in \mathcal{C}} \left\{ y^\top b + \frac{1}{2\lambda} \|Q_S^\top A^\top y\|_2^2 \right\}\,.
\end{align}
This approach corresponds to a \emph{left} sketch of a constrained quadratic program: its computational benefits have been extensively studied in the sketching literature and its statistical performance has been carefully analyzed in the case of oblivious embeddings (see, for instance,~\cite{pilanci2015randomized, lacotte2019faster, lacotte2020optimal}). Hence, our analysis extends the range of existing results of the left sketch approach to the case of adaptive embeddings.

\section{Numerical experiments}
\label{sec:numericalexperiments}

\textbf{Datasets.} We evaluate Algorithm~\ref{alg:MainAlgorithm} on the MNIST and CIFAR10 datasets with logistic regression. First, we aim to illustrate that the sketch size can be considerably smaller than the data dimension while recovering a close approximation to the optimal solution which achieves a similar test classification accuracy. Second, we aim to achieve significant computational speed-ups. To solve the primal program~\eqref{eqprimal}, we use two standard algorithms for empirical risk minimization, namely, stochastic gradient descent (SGD) with (best) fixed step size and stochastic variance reduction gradient (SVRG)~\cite{johnson2013accelerating} with (best) fixed step size and frequency update of the gradient correction. To solve the adaptive sketched program~\eqref{eqsketchedprimal}, we use SGD, SVRG and the sub-sampled Newton method~\cite{bollapragada2018exact, erdogdu2015convergence}: we refer to them as Sketch-SGD, Sketch-SVRG and Sketch-Newton. The latter is well-suited to the sketched program for a relatively small sketch size, as a low-dimensional Newton system can be efficiently solved at each iteration. For both datasets, we use $50000$ training and $10000$ testing images. We transform each image $a$ using random Fourier features $\psi(a) \in \real^d$, i.e., $\langle \psi(a), \psi(a^\prime) \rangle \approx \exp\left(-\gamma \|a-a^\prime\|^2_2 \right)$~\cite{rahimi2008random, scikit-learn}. For MNIST and CIFAR10, we choose respectively $d=10000$ and $\gamma=0.02$, and, $d=60000$ and $\gamma=0.002$, so that the primal programs are respectively $10000$-dimensional and $60000$-dimensional. Then, we train a classifier via a sequence of binary logistic regressions, using a one-versus-all procedure. 

\textbf{Adaptive Gaussian embeddings.} We evaluate the test classification error of the first-order estimator $\xone$. We solve to optimality the primal and sketched programs for values of $\lambda \in \{10^{-4}, 5 \cdot 10^{-5}, 10^{-5}, 5 \cdot 10^{-6}\}$ and sketch sizes $m \in \{128, 256, 512, 1024\}$. In Table~\ref{tab:rawrandomfourierfeatures} are reported the empirical means and standard deviations, which are averaged over $20$ trials. The adaptive sketched program yields a high accuracy classifier for most pairs $(\lambda, m)$: we match the best primal classifier with values of $m$ as small as $256$ for MNIST and $512$ for CIFAR10, which respectively corresponds to a dimension reduction by a factor $\approx 40$ and $\approx 120$. These results also suggest that adaptive sketching induces an implicit regularization effect, which is reminiscent of the benefits of spectral cutoff estimators~\cite{donoho1996neo}. For instance, on CIFAR10, using $\lambda=10^{-5}$ and $m=512$, we obtain an improvement in test accuracy by more than $2\%$ compared to $x^*$.

\begin{table}[!h]
\caption{Test classification error of adaptive Gaussian sketching on MNIST and CIFAR10 datasets. The subscript in $\xone_m$ refers to the sketch size $m$. Results are averaged over $20$ trials, and we report the empirical mean and standard deviation.}
  \label{tab:rawrandomfourierfeatures}
  \centering
  \begin{tabular}{|c|ccccc|}
    \cmidrule(r){1-6}
    $\lambda$ & $x^*_{\text{MNIST}}$ & $\xone_{128}$ & $\xone_{256}$ & $\xone_{512}$ & $\xone_{1024}$  \\
    \midrule
    $10^{-4}$   & $5.4$  & $4.8 \pm 0.2$ & $5.2 \pm 0.1$ & $5.3 \pm 0.1$ & $5.4 \pm 0.1$\\
    $5 \!\cdot\! 10^{-5}$ & $4.6$ & $3.8 \!\pm \!0.2$ & $4.0 \!\pm \!0.2$ & $4.3 \!\pm\! 0.1$ & $4.5\! \pm\! 0.1$\\
    $10^{-5}$   & $2.8$ &  $3.4 \pm 0.8$ & $2.4 \pm 0.2$ & $2.5 \pm 0.1$ & $2.8 \pm 0.1$ \\
    $5 \!\cdot\! 10^{-6}$ & $2.5$ & $4.9 \pm 2.1$ & $2.8 \pm 0.3$ & $2.6 \pm 0.2$ & $2.4 \pm 0.1$ \\
    \bottomrule
    $\lambda$ & $x^*_{\text{CIFAR}}$ & $\xone_{128}$ & $\xone_{256}$ & $\xone_{512}$ & $\xone_{1024}$\\
    \midrule
    $5\! \cdot\! 10^{-5}$ & $51.6$ & $50.5 \!\pm \!0.3$ & $50.6 \!\pm\! 0.3$ & $50.8\! \pm\! 0.2$ & $51.0 \!\pm \!0.2$ \\
    $10^{-5}$ & $48.2$ & $54.5 \pm 3.2$ & $47.7 \pm 0.6$ & $45.9 \pm 0.2$ & $46.2 \pm 0.2$ \\
    $5\! \cdot\! 10^{-6}$ & $47.6$ & $59.8 \pm 3.5$ & $51.9 \pm 2.1$ & $47.7 \pm 0.6$ & $45.8 \pm 0.6$ \\
    \bottomrule
    \end{tabular}
\end{table}

\textbf{Adaptive versus oblivious Gaussian sketching.} We evaluate the test classification error of two baseline estimators, that is, the first-order estimator $\xone_\dagger$ with oblivious Gaussian embeddings as proposed in~\cite{zhang2013recoveringCOLT, zhang2014recoveringIEEE} and described in Section~\ref{sectionunbiasedoblivioussketching}, and, the first-order estimator $\what{x}^{N}$ with adaptive Nystrom embeddings for which $S=A^\top \wtilde{S}$ with $\wtilde{S}$ a uniformly random column sub-sampling matrix. As reported in Table~\ref{tab:comparisonsketchingbaselines}, adaptive Gaussian sketching performs better for a wide range of values of sketch size $m$ and regularization parameter $\lambda$.

\begin{table}[!h]
\centering
    \caption{Test classification error (in percentage) on MNIST and CIFAR10. Results are averaged over $20$ trials and we report the empirical mean. The subscript under each estimator refers to the sketch size $m$.}
  \label{tab:comparisonsketchingbaselines}
  \begin{scriptsize}
  \begin{tabular}{|c|c c c c c c c|}
    \cmidrule(r){1-8}
    $\lambda$ & $x^*_{\text{MNIST}}$ & $\xone_{256}$ & $\xone_{1024}$ & $\xone_{\dagger, 256}$ & $\xone_{\dagger,1024}$ & $\what{x}^{N}_{256}$ & $\what{x}^{N}_{1024}$\\
    \midrule
    $5 \!\cdot\! 10^{-5}$ & 4.6 & 4.0 & 4.5 & 25.2 & 8.5 & 5.0 & 4.6 \\
    $5 \!\cdot\! 10^{-6}$ & 2.0 & 2.8  & 2.4 & 30.1 & 9.4 & 3.0 & 2.7 \\
    \bottomrule
    \end{tabular}
    \begin{tabular}{|c|c c c c c c c|}
    \cmidrule(r){1-8}
    $\lambda$ & $x^*_{\text{CIFAR}}$ & $\xone_{256}$ & $\xone_{1024}$ & $\xone_{\dagger, 256}$ & $\xone_{\dagger,1024}$ & $\what{x}^{N}_{256}$ & $\what{x}^{N}_{1024}$  \\
    \midrule
    $5 \!\cdot\! 10^{-5}$ & 51.6 & 50.6 & 51.0 & 88.2 & 70.5 & 55.8 & 53.1 \\
    $5 \!\cdot\! 10^{-6}$ & 47.6 & 51.9 & 45.8 & 88.9 & 80.1 & 57.2 & 55.8 \\
    \bottomrule
    \end{tabular}
    \end{scriptsize}
\end{table}

\textbf{Wall-clock time speed-ups.} For the first-order estimator $\xone$ with adaptive Gaussian embeddings, we compare the test classification error versus wall-clock time of the aforementioned optimization algorithms. Figure~\ref{fig:speed-ups} shows results for some values of $m$ and $\lambda$. We observe that solving instead the low-dimensional optimization problem offers significant speed-ups on the $10000$-dimensional MNIST problem, in particular for Sketch-SGD and for Sketch-SVRG, for which computing the gradient correction is relatively fast. Such speed-ups are even more significant on the $60000$-dimensional CIFAR10 problem, especially for Sketch-Newton, and a few iterations suffice to closely reach the approximate solution $\xone$, with a per-iteration time which is relatively small thanks to dimensionality reduction. Remarkably, it is more than $10$ times faster to reach the best test accuracy classifier using the sketched program. In addition to random Fourier features, we carry out another set of experiments with the CIFAR10 dataset, in which we pre-process the images. That is, similarly to~\cite{tu2016large, recht2018cifar}, we map each image through a random convolutional layer. Then, we kernelize these processed images using a Gaussian kernel with $\gamma=2 \cdot 10^{-5}$. Using our implementation, the best test accuracy of the kernel primal program~\eqref{eqprimalkernel} we obtained is $73.1\%$. Sketch-SGD, Sketch-SVRG and Sketch-Newton -- applied to the sketched kernel program~\eqref{eqsketchedprimalkernel} -- match this test accuracy, with significant speed-ups, as reported in Figure~\ref{fig:speed-ups}.

\begin{figure}[h!]
\centering
\includegraphics[width=0.95\columnwidth]{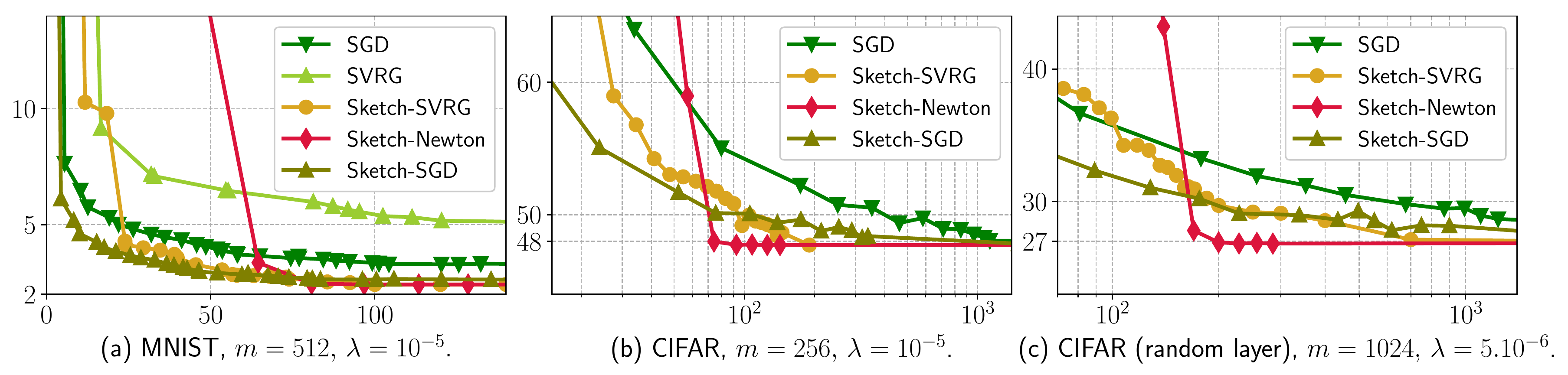}
\caption{Test classification error (percentage) versus wall-clock time (seconds).}
\label{fig:speed-ups}
\end{figure}

\textbf{Additional numerical details.} Experiments were run in Python on a workstation with $512$ GB of memory. We use our own implementation of each algorithm for a fair comparison. For SGD, we use a batch size equal to $128$. For SVRG, we use a batch size equal to $128$ and update the gradient correction every $400$ iterations. For Sketch-SGD, we use a batch size equal to $1024$. For Sketch-SVRG, we use a batch size equal to $64$ and update the gradient correction every $200$ iterations. Each iteration of the sub-sampled Newton method (Sketch-Newton) computes a full-batch gradient, and, the Hessian with respect to a batch of size $1500$. For SGD and SVRG, we considered step sizes $\eta$ between $10^{-2}$ and $10^{2}$. We obtained best performance for $\eta = 10^1$. For the sub-sampled Newton method, we use a step size $\eta=1$, except for the first $5$ iterations, for which we use $\eta=0.2$. In Figure~\ref{fig:speed-ups}, we did not report results for SVRG for solving the primal~\eqref{eqprimal} on CIFAR10, as the computation time for reaching a satisfying performance was significantly larger than for the other algorithms.

\section{Information theoretic lower-bounds and optimality for right-sketching}

We now consider the fundamental problem of estimating the mean of a random sample $b = A x_\text{pl} + w$ where $w \sim \mathcal{N}(0, \frac{\sigma^2}{n}\,I_n)$, under the assumption that the planted vector $x_\text{pl}$ satisfies $\|x_\text{pl}\|_2 \less 1$. Given an estimator $\what x$, we define its risk as $\mathfrak{R}(\what x) \defn \sup_{\|x_\text{pl}\|_2\less 1}\mathbb{E}_w\|A(\what x - x_\text{pl})\|_2^2$. A critical quantity to characterize the best achievable risk is the \emph{statistical dimension} $d_s$, defined as $d_s \defn \min\{k \gre 1 \mid \frac{\sigma^2 k}{n} \gre \sigma_{k+1}^2\}$. It satisfies the scaling $\frac{\sigma^2 d_s}{n} \asymp \sigma^2_{d_s+1}$. It is well-known (see, for instance,~\cite{bartlett2005local, geer2000empirical, yang2017randomized}) that the quantity $A^\top b$ is sufficient to construct an optimal estimator.

Information theoretic lower bounds for left-sketching via $SA, Sb$ were developed in \cite{pilanci2016iterative}. Surprisingly, it was shown that unless $m\asymp d$, left-sketching based estimators are sub-optimal. In turn, an iterative sketching method based on the sketches $\{S_iA,A^T(Ax_i- b)\}_{i=1}^T$ was shown to be statistically optimal for the broader class of constrained least squares problems, including unconstrained least squares, Lasso and nuclear norm constrained problems.

In our context, we are given a sketch $AS$ where $S \in \real^{d \times m}$. Note that right-sketched based optimization problem \eqref{eqsketchedprimal} for the least squares objective is of the form
\begin{align}
    \min_\alpha \|AS\alpha -b\|_2^2 + \phi(\alpha) = \min_{\alpha} \|AS\alpha\|_2^2 - 2 \alpha^T S^T A^T b+ \|b\|_2^2 + \phi(\alpha)
\end{align}
where $\phi(\alpha)$ is an arbitrary regularization term. The preceding line shows that right-sketching estimators are functions of $S^TA^Tb$. Therefore, bounds on the mutual information between $S^T A^T b$ and  $x_\text{pl}$ can be leveraged to develop information theoretic lower bounds under this observation model. Consequently, we consider right-sketching estimators based on the observation $S^\top A^\top b$, and we aim to characterize the minimax risk $\mathfrak{M}_S \defn \inf_{\what x} \mathfrak{R}(\what x)$ where the infimum is taken over estimators $\what x \equiv \what x (S^\top A^\top b)$. We aim to show that for an oblivious Gaussian embedding $S \in \real^{d \times m}$, a sketch size $m \asymp d_s$ and polynomial or exponential decays, it holds with high-probability that
\begin{align}
    \mathfrak{M}_S \asymp \frac{\sigma^2 d_s}{n}\,.
\end{align}
Moreover, this minimax lower-bound is achieved by the right-sketching estimator. 
Our proof of the lower-bound is based on the standard local Fano method.
\begin{theorem}
\label{theoremlowerboundleastsquares}
Let $m \gre 1$ be a sketch size and $S \in \real^{d \times m}$ be a Gaussian embedding, and assume that $d_s \gre 4$. Conditional on the event $\|P_{AS}^\perp A\|_2^2 \less \frac{\sigma_{d_s+1}^2}{2}$, it holds that
\begin{align}
\label{eqnlowerboundminimaxrate}
    \mathfrak{M}_S \gre c_0 \cdot \sigma_{d_s+1}^2 \asymp c_0 \cdot \frac{\sigma^2 d_s}{n}\,,
\end{align}
where $c_0$ is a universal constant such that $c_0 \gre \frac{1}{4096}$. Consequently, for a polynomial decay $\sigma_j = j^{-\frac{1+\nu}{2}}$ with $\nu > 0$ for which $d_s \asymp \left(\frac{n}{\sigma^2}\right)^{\frac{1}{2+\nu}}$ and for a sketch size $m \gre c^{\text{poly}}_{\nu} \cdot d_s$ where $c^{\text{poly}}_\nu \asymp (\nu^{-1}+1)^{(1+\nu)}$, it holds with probability at least $1-6 e^{-d_s}$ that 
\begin{align}
    \mathfrak{M}_S \gre c_0 \cdot \sigma_{d_s+1}^2 \asymp c_0 \cdot \left(\frac{\sigma^2}{n}\right)^{\frac{1+\nu}{2+\nu}}\,.
\end{align}
For an exponential decay $\sigma_j = e^{-\frac{\nu j}{2}}$ with $\nu > 0$ such that $n \nu/\sigma^2 \gre 5408$ for which $d_s \asymp \frac{1}{\nu}\log(n/\sigma^2)$ and for a sketch size $m \gre 4d_s$, it holds with probability at least $1-6e^{-d_s}$ that
\begin{align}
    \mathfrak{M}_S \gre c_0 \cdot \sigma_{d_s+1}^2 \asymp c_0 \cdot \frac{\sigma^2 \log(n/\sigma^2)}{\nu n}\,.
\end{align}
\end{theorem}
We show next that the risk of the zero-order estimator $\xzero$ with sketch size $m \asymp d_s$ achieves the above lower bound on the minimax rate of estimation~\eqref{eqnlowerboundminimaxrate}, as the regularization parameter $\lambda$ goes to $0$.
\begin{corollary}
\label{corollaryoptimalityzeroorder}
Let $m \gre 1$ be a sketch size and $S \in \real^{d \times m}$ be a Gaussian embedding. Assume that $d_s \gre 4$. For a polynomial decay $\sigma_j = j^{-\frac{1+\nu}{2}}$ and with a sketch size $m = c^{\text{poly}}_{\nu} \cdot d_s$ where $c^{\text{poly}}_\nu \asymp (\nu^{-1}+1)^{(1+\nu)}$, we have with probability at least $1-6e^{-d_s}$ that
\begin{align}
    \lim_{\lambda \to 0} \mathfrak{R}(\xzero) \less \left(c^{\text{poly}}_{\nu} + \frac{1}{2}\right) \cdot \frac{\sigma^2 d_s}{n}\,.
\end{align}
For an exponential decay $\sigma_j = e^{-\frac{\nu j}{2}}$ with $\nu > 0$ such that $n \nu/\sigma^2 \gre 5408$, and with a sketch size $m = 4d_s$, we have with probability at least $1-6e^{-d_s}$ that 
\begin{align}
    \lim_{\lambda \to 0} \mathfrak{R}(\xzero) \less 5 \cdot \frac{\sigma^2 d_s}{n}\,.
\end{align}
\end{corollary}

Therefore we conclude that right-sketching is minimax optimal unlike left-sketching under this standard observation model. This result complements the left-sketching lower-bounds from \cite{pilanci2016iterative} and indicates that right-sketching is more advantageous for problems with small statistical dimension.

\section{Conclusion}

Through tighter performance bounds and analytical comparison, we have shown that the dual reconstruction method along with adaptive embeddings yields an estimator $\xone$ which significant improves over the linear reconstruction map and oblivious sketching in the context of convex smooth optimization, as usually considered in the literature. Furthermore, we have extended this method to non-smooth optimization problems, and our method requires solving an additional dual optimization problem with potentially very few variables: in contrast to optimizing over a random subset of dual variables (e.g., SDNA), our primal low-dimensional approach selects the appropriate subset of dual variables. Most of our results mirror those established for left-sketching methods~\cite{pilanci2015randomized}, although they are fundamentally different due to the choice of the adaptive embedding and thus, require a novel analysis technique.

\appendices

\section{Analysis of adaptive sketches}

\subsection{Proof of Lemma~\ref{lemmaupperboundZf}}

A proof of an upper bound on the singular value $\|P_SA^\top\|_2$ is provided in~\cite{witten2015randomized}, and our analysis is an adaptation of this proof to the restricted case. For two real-valued random variables $X$ and $Y$, we say that $X$ is stochastically dominated by $Y$ if $\mathbb{P}(X \gre \tau) \less \mathbb{P}(Y \gre \tau)$ for any $\tau \in \mathbb{R}$, and we write $X \overset{d}{\less} Y$. 

Let $S \in \real^{n\times m}$ be a matrix with i.i.d.~Gaussian entries $\mathcal{N}(0,1/m)$. We use the notation $f(A,S) \defn P_{A^\top S}^\perp A^\top$, and we introduce a thin SVD of $A$, denoted by $A= U\Sigma V^\top$, where $\Sigma = \text{diag}\{\sigma_1, \dots, \sigma_\rho\}$. Note that 
\begin{align*}
    P_{A^\top S} = A^\top S (S^\top A A^\top S)^\dagger S^\top A = V \Sigma U^\top S (S^\top U \Sigma \underbrace{V^\top V}_{=\,I_\rho} \Sigma U^\top S)^\dagger S^\top U\Sigma V^\top = V P_{\Sigma U^\top S} V^\top\,.
\end{align*}
Consequently, we have that
\begin{align*}
    f(A, S) = (I-P_{A^\top S}) A^\top = (I - V P_{\Sigma U^\top S} V^\top) V \Sigma U^\top &= V \Sigma U^\top - V P_{\Sigma U^\top S} \underbrace{V^\top V}_{=\,I_\rho} \Sigma U^\top
    = V \underbrace{(\Sigma - P_{\Sigma U^\top S} \Sigma)}_{= \, f(\Sigma, U^\top S)} U^\top\,.
\end{align*}
That is, $f(A, S) = V f(\Sigma, U^\top S) U^\top$. Let $G \in \real^{\rho \times m}$ be a matrix with i.i.d.~Gaussian entries $\mathcal{N}(0,1/m)$. By rotational invariance of the Gaussian distribution, it holds that $G \overset{\mathrm{d}}{=}U^\top S$. Therefore, $f(A,S) \overset{\mathrm{d}}{=} V f(\Sigma, G) U^\top$. Since $V$ is an isometry (i.e., $\|Vw\|_2 = \|w\|_2$ for any $w \in \real^\rho$), it follows that
\begin{align*}
    \mathcal{Z}_f \overset{\text{d}}{=} \sup_{\Delta \in \mathcal{C}_{z^*}} {\|f(\Sigma, G) U^\top \Delta\|}_2\,. 
\end{align*}
For $t > 0$, we define $M(t) \defn \begin{bmatrix} t I_k & 0 \\ 0 & \Sigma_{\rho-k} \end{bmatrix}$. Note that $\Sigma \preceq M(t)$ for any $t \gre \sigma_1$. According to Lemma~2.5 in~\cite{witten2015randomized}, it holds that ${\|f(\Sigma_1,G) w\|}_2 \less {\|f(\Sigma_2, G)w\|}_2$ for any $w \in \real^\rho$ and any two positive definite diagonal matrices $\Sigma_1$ and $\Sigma_2$ such that $\Sigma_1 \preceq \Sigma_2$. As a consequence, for any $\Delta \in \mathcal{C}_{z^*}$, it holds almost surely that $\|f(\Sigma, G)U^\top \Delta \|_2 \less \lim_{t \to \infty}\,\|f(M(t), G)U^\top \Delta \|_2$, and thus,
\begin{align}
\label{eqndupperboundZf}
    \mathcal{Z}_f \overset{d}{\less} \sup_{\Delta \in \mathcal{C}_{z^*}} \lim_{t \to \infty}\, {\|f(M(t), G) U^\top \Delta\|}_2\,.
\end{align}
The following fact has already been proved in~\cite{witten2015randomized}: it holds that 
\begin{align*}
    \lim_{t \to + \infty} f(M(t), G) \, \overset{\mathrm{d}}{=} \, \begin{bmatrix} 0 & 0 \\ f(\Sigma_{\rho-k}, X_2) X_1 \Lambda^{-1} \Omega & f(\Sigma_{\rho-k}, X_2) \end{bmatrix}\,,
\end{align*}
where $X_1 \in \mathbb{R}^{(\rho-k)\times k}$, $X_2 \in \mathbb{R}^{(\rho-k) \times k}$, $\Lambda \in \mathbb{R}^{k \times k}$ and $\Omega \in \mathbb{R}^{k \times k}$ are independent random matrices, and, $X_1$ and $X_2$ have independent i.i.d.~Gaussian entries, $\Lambda$ is diagonal with entries distributed as the first $k$ singular values of a $k \times m$ Gaussian matrix, and $\Omega$ is an orthogonal matrix. Plugging-in this limit in the right-hand side of~\eqref{eqndupperboundZf}, we obtain
\begin{align*}
    \mathcal{Z}_f \overset{\mathrm{d}}{\less} \sup_{\Delta \in \mathcal{C}_{z^*}} \|f(\Sigma_{\rho-k}, X_2) X_1 \Lambda^{-1}\Omega U_k^\top \Delta + f(\Sigma_{\rho-k}, X_2) U_{\rho-k}^\top \Delta\|_2\,.
\end{align*}
By triangular inequality, it follows that
\begin{align*}
    \mathcal{Z}_f &\overset{\mathrm{d}}{\less} \sup_{\Delta \in \mathcal{C}_{z^*}} \|f(\Sigma_{\rho-k}, X_2) X_1 \Lambda^{-1} \Omega U_k^\top \Delta\|_2 + \sup_{\Delta \in \mathcal{C}_{z^*}} \|f(\Sigma_{\rho-k}, X_2)U_{\rho-k}^\top \Delta\|_2 \\
    &\less \text{rad}\{U_k^\top \mathcal{C}_{z^*}\} \cdot \|f(\Sigma_{\rho-k}, X_2) X_1 \Lambda^{-1} \Omega\|_2 + \sup_{\Delta \in \mathcal{C}_{z^*}} \, \|f(\Sigma_{\rho-k}, X_2)U_{\rho-k}^\top \Delta\|_2\\
    & \less \text{rad}\{U_k^\top \mathcal{C}_{z^*}\} \cdot \|\Sigma_{\rho-k} X_1 \Lambda^{-1}\|_2 + \sup_{\Delta \in \mathcal{C}_{z^*}} \, \|\Sigma_{\rho-k} U_{\rho-k}^\top \Delta\|_2\,,
\end{align*}
and we used the fact that $\|f(\Sigma_{\rho-k}, X_2) W\|_2 = \|P_{\Sigma_{\rho-k}X_2}^\perp \Sigma_{\rho-k}W\|_2 \less \|\Sigma_{\rho-k} W\|_2$ for any arbitrary matrix $W$.
According to Corollary 10.9~\cite{halko2011}, it holds with probability at least $1-6e^{-k}$ that 
\begin{align*}
    \|\Sigma_{\rho-k} X_1 \Lambda^{-1}\|_2 \less 17 \sqrt{2} \cdot \sigma_{k+1} + \frac{9}{\sqrt{k}} \cdot \sqrt{\sum_{j > k}\sigma_j^2} \less \underbrace{\max\{17 \sqrt{2}, 9\}}_{\less\,25} \cdot \left(\sigma_{k+1} + \sqrt{\frac{1}{k}\sum_{j > k}\sigma_j^2}\right)
    \less 25 \cdot R_k(A)\,.
\end{align*}
In summary, we have shown that with probability at least $1-6e^{-k}$, we have
\begin{align*}
    \mathcal{Z}_f \less 25\cdot  \text{rad}\{U_k^\top \mathcal{C}_{z^*}\} \cdot R_k(A) + \sup_{\Delta \in \mathcal{C}_{z^*}} \|\Sigma_{\rho-k}U_{\rho-k}^\top \Delta\|_2\,,
\end{align*}
which is the first part of the claim. Furthermore, since $\mathcal{C}_{z^*} \subset \mathcal{B}_2^n$, we always have that
\begin{align*}
    \text{rad}\{U_k^\top \mathcal{C}_{z^*}\} \less 1\,,\quad \text{and} \quad  \sup_{\Delta \in \mathcal{C}_{z^*}} \|\Sigma_{\rho-k}U_{\rho-k}^\top \Delta\|_2 \less \sigma_{k+1}\,.
\end{align*}
Therefore, $\mathcal{Z}_f \less 25 \cdot R_k(A) + \sigma_{k+1} \less 26 \cdot R_k(A)$ with probability at least $1-6e^{-k}$. This holds in particular when $\mathcal{C}_{z^*} = \mathcal{B}_2^n$, which concludes the proof.
\qed

\subsection{Proof of Lemma~\ref{lemmarandomizedupperboundZd}}
\label{prooflemmarandomizedupperboundZf}

From Lemma~\ref{lemmaupperboundZf}, we know that with $m=2k$ and $S=A^\top \wtilde S$ where $\wtilde S \in \real^{n\times m}$ has i.i.d.~Gaussian entries, it holds with probability at least $1-6e^{-k}$ that
\begin{align*}
    \mathcal{Z}_f \less 25 \cdot \textrm{rad}(U_k^\top \mathcal{C}_{z^*}) \cdot R_k(A) + \sup_{\Delta \in \mathcal{C}_{z^*}} {\|\Sigma_{\rho-k}U_{\rho-k}^\top \Delta\|}_2\,,
\end{align*}
From Theorem~7.7.1.~in~\cite{vershynin2018high} and using that $\textrm{rad}(\mathcal{C}_{z^*}) \less 1$, we have
\begin{align}
    \textrm{rad}\!\left(U_k^\top \mathcal{C}_{z^*}\right) \less c^\prime_1 \cdot \frac{\omega(\mathcal{C}_{z^*}) + \sqrt{m}}{\sqrt{n}}\,,
\end{align}
with probability at least $1-2e^{-k}$, for some universal constant $c^\prime_1 > 0$. It remains to control the term $\sup_{\Delta \in \mathcal{C}_{z^*}} {\|\Sigma_{\rho-k}U_{\rho-k}^\top \Delta\|}_2 = \sup_{\Delta \in \mathcal{C}_{z^*}, t \in \mathcal{E}} \, \langle t, U_{\rho-k}^\top \Delta\rangle$, where we define the ellipsoid $\mathcal{E} = \left\{\Sigma_{\rho-k} z \mid {\|z\|}_2 \less 1 \right\}$. We use again a Chevet-type inequality~\cite{adamczak2011chevet}, which yields that
\begin{align}
    \sup_{\Delta \in \mathcal{C}_{z^*}, t \in \mathcal{E}} \, \langle t, U_{\rho-k}^\top \Delta\rangle  \less c_2^\prime \cdot \frac{1}{\sqrt{n}} \left(\omega(\mathcal{C}_{z^*})\,\textrm{rad}(\mathcal{E}) + w(\mathcal{E}) \, \textrm{rad}(\mathcal{C}_{z^*})\right)\,,
\end{align}
with probability at least $1-c_4 e^{-c_5(\rho-k)}$ for some universal constants $c_4, c_5 > 0$, and where we introduced the ellipsoid $\mathcal{E} = \left\{\Sigma_{\rho-k} z \mid {\|z\|}_2 \less 1 \right\}$. Using the facts that $\omega(\mathcal{E}) \less \left(\sum_{j=k+1} \sigma_j^2\right)^{\frac{1}{2}}$, $\textrm{rad}(\mathcal{E}) = \sigma_{k+1}$ and $\textrm{rad}(\mathcal{C}_{z^*}) \less 1$, the above inequality becomes
\begin{align*}
    \sup_{\Delta \in \mathcal{C}_{z^*}} {\|\Sigma_{\rho-k}U_{\rho-k}^\top \Delta\|}_2 &\less c_2^\prime \left(\sigma_{k+1} \frac{\omega(\mathcal{C}_{z^*})}{\sqrt{n}} + \sqrt{\frac{k}{n}} \cdot \sqrt{ \frac{1}{k} \sum_{j=k+1}^r \sigma_j^2}\right)\\
    & \less c_3^\prime \cdot \frac{\omega(\mathcal{C}_{z^*}) + \sqrt{m}}{\sqrt{n}} \cdot R_k(A)\,.
\end{align*}
with probability at least $1-c_4 e^{-c_5 (\rho-k)}$, for some universal constant $c_3^\prime > 0$. By union bound, we obtain the claimed result.

\subsection{Proof of Lemma~\ref{lemmaresidualsrht}}
We use Theorem 2.1 from~\cite{boutsidis2013}, which states the following. Given a target rank $2 \less k \less \rho$ and a failure probability $0 < \delta < 1$, let $S=A^\top \wtilde S$ where $\wtilde S \in \real^{n \times m}$ is a SRHT with $n \gre m \gre 19 (\sqrt{k} + \sqrt{8 \log(n/\delta)})^2 \log(k/\delta)$. Then, it holds with probability at least $1-5\delta$ that ${\|P_S^\perp A^\top\|}_2 \less \left(4 + \sqrt{\frac{3 \log(n/\delta) \log(\rho/\delta)}{m}}\right) \cdot \sigma_{k+1} + \sqrt{\frac{3 \log(\rho/\delta)}{m}} \cdot \sqrt{\sum_{j=k+1}^\rho \sigma_j^2}$. Picking $\delta = \frac{1}{n}$ and using that $n \gre \rho$, we obtain for $m \gre 19 \left(\sqrt{k} + 4 \sqrt{\log n}\right)^2 \log(nk)$ and with probability at least $1-\frac{5}{n}$ that ${\|P_S^\perp A^\top\|}_2 \less \Big(4 + \sqrt{\frac{12 \log^2(n)}{m}}\Big) \cdot \sigma_{k+1} + \sqrt{6} \cdot \sqrt{\frac{\log(n)}{m}} \cdot \sqrt{\sum_{j=k+1}^\rho \sigma_j^2}$. We conclude by using that $\sqrt{\frac{12 \log^2(n)}{m}} \less 1$ and $\sqrt{\frac{\log(n)}{m}} \less \frac{1}{\sqrt{k}}$.
\qed

\subsection{Proof of Theorem~\ref{theoremzeroversusfirstorder}}
\label{prooftheoremzeroversusfirstorder}

From Propositions~\ref{propositionfenchelduality} and~\ref{propositionsketchedfenchelduality}, we know that there exist $g_{z^*} \in \partial f^*(z^*)$ and $g_{y^*} \in \partial f^*(y^*)$ such that $g_{z^*} + \frac{1}{\lambda} AA^\top z^* = 0$ and $g_{y^*} + \frac{1}{\lambda} A P_S A^\top y^* = 0$. We define the error vector $\Delta \defn y^* - z^*$, which belongs to the tangent cone $\mathcal{T}_{z^*}$. Subtracting the two previous equalities, we obtain $g_{y^*} - g_{z^*} + \frac{1}{\lambda} A P_S A^\top \Delta = \frac{1}{\lambda} A P_S^\perp A^\top z^*$. Multiplying by $\lambda \Delta^\top$ and using the fact that $P_S^2 = P_S$, it follows that 
\begin{align}
\label{eqnintermediateineqtheorem1}
    \lambda \, \langle \Delta, g_{y^*} - g_{z^*} \rangle + \|P_S A^\top \Delta\|_2^2 = \langle \Delta, A P_S^\perp A^\top z^* \rangle\,.
\end{align}
Smoothness of $f$ implies that its Fenchel conjugate $f^*$ is $\mu^{-1}$-strongly convex, i.e., $\langle \Delta, g_{y^*} - g_{z^*} \rangle \gre \frac{1}{\mu} \|\Delta\|_2^2$. By definition of $\mathcal{Z}_f$, we have that $\|P_S^\perp A^\top \Delta\|_2 \less \mathcal{Z}_f \|\Delta\|_2$, and consequently,
\begin{align*}
    \|P_S A^\top \Delta\|_2^2 = \|A^\top \Delta\|_2^2 - \|P_S^\perp A^\top \Delta\|_2^2 \gre \|A^\top \Delta\|_2^2 - \mathcal{Z}_f^2 \|\Delta\|_2^2\,.
\end{align*}
Plugging-in the previous inequalities into~\eqref{eqnintermediateineqtheorem1}, we obtain
\begin{align*}
    \left(\frac{\lambda}{\mu} - \mathcal{Z}_f^2 \right) \cdot \|\Delta\|_2^2 + \|A^\top \Delta\|_2^2 \less \langle \Delta, AP_S^\perp A^\top z^* \rangle\,.
\end{align*}
Using the assumption $\lambda \gre 2\mu \mathcal{Z}_f^2$, it follows that $\frac{\lambda}{\mu} - \mathcal{Z}_f^2 \gre \frac{\lambda}{2\mu}$. By Cauchy-Schwarz inequality, we have $|\langle \Delta, AP_S^\perp A^\top z^* \rangle| \less \|P_S^\perp A^\top \Delta\|_2 \|P_S^\perp A^\top z^*\|_2$. Hence, we obtain the inequality
\begin{align*}
    \frac{\lambda}{2\mu}\|\Delta\|_2^2 + \|A^\top \Delta\|_2^2 \less \|P_S^\perp A^\top \Delta\|_2 \|P_S^\perp A^\top z^*\|_2\,.
\end{align*}
Using the identity $a^2 + b^2 \gre 2 ab$ with $a = \sqrt{\frac{\lambda}{2\mu}}\|\Delta\|_2$ and $b=\|A^\top \Delta\|_2$ and the inequality $\|P_S^\perp A^\top \Delta\|_2 \less \mathcal{Z}_f \cdot \|\Delta\|_2$, it follows that
\begin{align*}
    \sqrt{\frac{2\lambda}{\mu}} \cdot \|\Delta\|_2 \cdot \|A^\top \Delta\|_2 \less \|\Delta\|_2 \cdot \mathcal{Z}_f \cdot \|P_S^\perp A^\top z^*\|_2\,.
\end{align*}
Using the identities $x^* = -\lambda^{-1} A^\top z^*$ and $\xone = -\lambda^{-1} A^\top y^*$ and rearranging the above inequality, we obtain
\begin{align}
\label{eqnintermediateineqtheorem2}
    \|\xone - x^*\|_2 \less \sqrt{\frac{\mu}{2\lambda}} \cdot \mathcal{Z}_f \cdot \|P_S^\perp x^*\|_2\,.
\end{align}
It always holds that $\|P_S^\perp x^*\|_2 \less \|x^*\|_2$, and consequently, we have $\frac{\|\xone - x^*\|_2}{\|x^*\|_2} \less \sqrt{\frac{\mu}{2\lambda}} \cdot \mathcal{Z}_f$. On the other hand, we have ${\|\xzero - x^*\|}_2^2 = {\|P_S(\xzero - x^*)\|}_2^2 + {\|P_S^\perp x^*\|}^2_2$, which further implies ${\|P_S^\perp x^*\|}_2 \less {\|\xzero - x^*\|}_2$. Consequently, we also obtain from~\eqref{eqnintermediateineqtheorem2} that $\frac{\|\xone - x^*\|_2}{\|x^*\|_2} \less \sqrt{\frac{\mu}{2\lambda}} \cdot \mathcal{Z}_f \cdot \frac{\|\xzero - x^*\|_2}{\|x^*\|_2}$, and this concludes the proof.
\qed

\subsection{Proof of Theorem~\ref{theoremiterativemethod}}

Using an induction argument, it suffices to show that for any $t \gre 0$ and provided that $\lambda \gre 2\mu \mathcal{Z}_f^2$, we have
\begin{align}
\label{eqncontractioninequality}
    {\|\xone_{t+1} - x^*\|}_2 \less \sqrt{\frac{\mu}{2\lambda}} \cdot \mathcal{Z}_f \cdot {\|\xone_t - x^*\|}_2.
\end{align}
It should be noted that for $t=0$, since $\xone_0 = 0$, the latter inequality is exactly the regret bound~\eqref{eqnboundfirstorder}. The proof for any $t \gre 0$ follows steps similar to the proof of Theorem~\ref{theoremzeroversusfirstorder}. Fix $t \gre 0$. The Fenchel dual of the sketched program~\eqref{eqnintermediatesketchedprimal} is given by 
\begin{align*}
    \min_{y \in \real^n} \left\{f^*(y) - y^\top A P_S^\perp \xone_t + \frac{1}{2\lambda} \|P_S A^\top y\|_2^2\right\}\,.
\end{align*}
Using arguments similar to the proof of Proposition~\ref{propositionsketchedfenchelduality}, we obtain that there exists a dual solution $y^*_t \in \real^n$, and that $\alpha^*_{\dagger,t}$ and $y^*_t$ are related through the KKT conditions $y^*_t = \nabla f(AQ_S \alpha^*_{\dagger,t} + A\xone_t)$. Recall that, by definition, $\xone_{t+1} = -\lambda^{-1}A^\top \nabla f(AQ_S \alpha^*_{\dagger,t} + A\xone_t)$, i.e., $\xone_{t+1} = -\lambda^{-1} A^\top y^*_t$. We define the error vector $\Delta \defn y^*_t - z^*$, which belongs to the tangent cone $\mathcal{T}_{z^*}$. By first-order optimality conditions of $y^*_t$ and $z^*$, we know that there exist $g_{y^*_t} \in \partial f^*(y^*_t)$ and $g_{z^*} \in \partial f^*(z^*)$ such that $g_{y^*_t} + \frac{1}{\lambda} A P_S A^\top y^*_t - AP_S^\perp \xone_t = 0$ and $g_{z^*} + \frac{1}{\lambda} AA^\top z^* = 0$. Subtracting both sides of the previous inequalities and multiplying by $\lambda \Delta^\top$, we obtain that 
\begin{align}
\label{eqnintermediateiterative1}
    \lambda \, \langle \Delta, g_{y^*_t} - g_{z^*} \rangle + \|P_S A^\top \Delta\|_2^2 = \langle A P_S^\perp (\lambda \xone_t + A^\top z^*), \Delta \rangle\,.
\end{align}
By definition of $\mathcal{Z}_f$, we have that $\|P_S^\perp A^\top \|_2 \less \mathcal{Z}_f \|\Delta\|_2$, and thus, $\|P_S A^\top \Delta\|_2^2 \gre \|A^\top \Delta\|_2^2 - \mathcal{Z}_f^2 \|\Delta\|_2^2$. Plugging-in this inequality into~\eqref{eqnintermediateiterative1} as well as the strong convexity inequality $\langle \Delta, g_{y^*_t} - g_{z^*} \rangle \gre \frac{1}{\mu} \|\Delta\|_2^2$, we obtain that
\begin{align*}
    \left(\frac{\lambda}{\mu} - \mathcal{Z}_f^2\right) \|\Delta\|_2^2 + \|A^\top \Delta\|_2^2 \less \langle \Delta, AP_S^\top (\lambda \xone_t+A^\top z^*) \rangle\,.
\end{align*}
Using the assumption $\lambda \gre 2 \mu \mathcal{Z}_f^2$, we get that $\frac{\lambda}{\mu} - \mathcal{Z}_f^2 \gre \frac{\lambda}{2\mu}$. Using further the identity $a^2 + b^2 \gre 2 ab$ with $a = \sqrt{\frac{\lambda}{2\mu}}\|\Delta\|_2$ and $b=\|A^\top \Delta\|_2$, we deduce that
\begin{align*}
    \sqrt{\frac{2\lambda}{\mu}} \cdot \|\Delta\|_2 \cdot \|A^\top \Delta\|_2 \less \mathcal{Z}_f \cdot \|\Delta\|_2 \cdot \|\lambda \xone_t + A^\top z^*\|_2\,.
\end{align*}
Using the identities $\xone_{t+1} = -\lambda^{-1} A^\top y^*_t$, $x^* = -\lambda^{-1} A^\top z^*$ and thus $A^\top \Delta = \lambda (x^* - \xone_{t+1})$, we finally obtain
\begin{align*}
    \|\xone_{t+1}-x^*\|_2 \less \sqrt{\frac{\mu}{2\lambda}} \cdot \mathcal{Z}_f \cdot \|\xone_t - x^*\|_2\,,
\end{align*}
which is the claimed inequality~\eqref{eqncontractioninequality}, and this concludes the proof.

\subsection{Proof of Theorem~\ref{theoremnonsmooth}}

Using the same arguments as in the proofs of Propositions~\ref{propositionfenchelduality} and~\ref{propositionsketchedfenchelduality}, we obtain that there exist dual solutions $z^*$ and $y^*$ which belong to the image of $\partial f$. The function $f$ is $L$-Lipschitz, and this implies that $\|z^*\|_2 \less L$ and $\|y^*\|_2 \less L$. By first-order optimality conditions, there exist subgradients $g_{y^*} \in \partial f^*(y^*)$ and $g_{z^*} \in \partial f^*(z^*)$ such that $g_{y^*} + \frac{1}{\lambda} AP_SA^\top y^* = 0$ and $g_{z^*} + \frac{1}{\lambda} A A^\top z^* = 0$. We define the error vector $\Delta \defn y^* - z^*$, which belongs to the tangent cone $\mathcal{T}_{z^*}$ and which satisfies $\|\Delta\|_2 \less 2L$. Subtracting the first-order optimality conditions on $y^*$ and $z^*$, and multiplying by $\Delta^\top$, we obtain that $\langle \Delta, g_{y^*} - g_{z^*}\rangle + \lambda^{-1}\langle \Delta, A P_S A^\top \Delta \rangle = \lambda^{-1} \langle \Delta, A P_S^\perp A^\top z^* \rangle$. By convexity of $f$, we have $\langle \Delta, g_{y^*} - g_{z^*}\rangle \gre 0$. Using $|\langle \Delta, AP_S^\perp A^\top z^*\rangle| \less \|P_S^\perp A^\top \Delta\|_2 \|P_S^\perp A^\top z^*\|_2$, we further obtain that 
\begin{align}
\label{eqnintermediateineqnonsmooth}
    \|A^\top \Delta\|_2^2 \less \|P_S^\perp A^\top \Delta\|_2^2 + \|P_S^\perp A^\top \Delta\|_2 \|P_S^\perp A^\top z^*\|_2\,.
\end{align}
Since $\Delta \in \mathcal{T}_{z^*}$, we have $\|P_S^\perp A^\top \Delta\|_2 \less \|\Delta\|_2 \mathcal{Z}_f \less 2L \mathcal{Z}_f $. Moreover, we have $\|P_S^\perp A^\top z^*\|_2 \less \|z^*\|_2 \|P_S^\perp A^\top\|_2 \less L \|P_S^\perp A^\top\|_2$. Combining the two latter inequalities with~\eqref{eqnintermediateineqnonsmooth}, we find that $\|A^\top \Delta\|_2^2 \less 4 L^2 \mathcal{Z}_f^2 + 2 L^2 \mathcal{Z}_f \|P_S^\perp A^\top\|_2$. Dividing by $\lambda^2$ and using the identities $x^* = -\lambda^{-1}A^\top z^*$ and $\xone = - \lambda^{-1} A^\top y^*$, we obtain the claimed inequality
\begin{align*}
    \|\xone - x^*\|_2 \less 2 \cdot \frac{L}{\lambda} \cdot \sqrt{\mathcal{Z}_f^2 + \frac{1}{2}\,\mathcal{Z}_f \|P_S^\perp A^\top\|_2}\,.
\end{align*}
On the other hand, under the assumption that $\inf_w f(w) > -\infty$, it holds that $0 \in \text{dom} f^*$. This implies that $-z^* \in \mathcal{T}_{z^*}$, and consequently, $\|P_S^\perp A^\top z^*\|_2 \less L \mathcal{Z}_f$. Since $\Delta \in \mathcal{T}_{z^*}$, we have $\|P_S^\perp A^\top \Delta\|_2 \less 2L \mathcal{Z}_f$. Combining the two latter inequalities with~\eqref{eqnintermediateineqnonsmooth}, we obtain the refined inequality
\begin{align*}
    \|\xone - x^*\|_2 \less \sqrt{6} \cdot \frac{L}{\lambda} \cdot \mathcal{Z}_f\,,
\end{align*}
and this concludes the proof.
\qed

\section{Analysis of oblivious sketches}

\subsection{Proof of Lemma~\ref{lemmaupperboundZfoblivious}}

Let $S \in \real^{d \times m}$ be a matrix with i.i.d.~Gaussian entries. By rotational invariance of the Gaussian distribution, there exists a random Haar matrix $Q \in \real^{d \times (d-m)}$ such that $P_S^\perp = Q Q^\top$. We have $\mathcal{Z}_f = \sup_{\Delta \in \mathcal{C}_{z^*}}\|P_S^\perp A^\top \Delta\|_2 = \sup_{\Delta \in \mathcal{C}_{z^*}} \|QQ^\top A^\top \Delta\|_2 = \sup_{\Delta \in \mathcal{C}_{z^*}} \|Q^\top A^\top \Delta\|_2$, where the last equality holds due to the fact that $Q$ is an isometry. According to standard concentration bounds (see, for instance, Theorem~7.7.1 in \cite{vershynin2018high}), it holds with probability at least $1-2e^{-(d-m)}$ that 
\begin{align*}
    \sup_{\Delta, \Delta^\prime \in \mathcal{C}_{z^*}} \|Q^\top A^\top (\Delta-\Delta^\prime)\|_2 \less c \cdot \left(\frac{\omega(A^\top \mathcal{C}_{z^*})}{\sqrt{d}} + \sqrt{\frac{d-m}{d}} \cdot \|A\|_2\right),
\end{align*}
where $c > 0$ is some universal constant. The claimed result follows from the fact that $0 \in \mathcal{C}_{z^*}$ and thus, 
\begin{align*}
    \sup_{\Delta \in \mathcal{C}_{z^*}}\|Q^\top A^\top \Delta\|_2 \less \sup_{\Delta, \Delta^\prime \in \mathcal{C}_{z^*}} \|Q^\top A^\top (\Delta-\Delta^\prime)\|_2\,.
\end{align*}
\qed

\subsection{Proof of Theorem~\ref{theoremlowerboundfirstorderobliviousrandom}}

\subsubsection{Technical preliminaries}
\begin{lemma}
\label{lemmaexistencesymmetricmappingsmoothness}
Suppose that $x, y \in \real^n$ such that $\beta\, x^\top y > \|y\|^2$ for some $\beta > 0$. Then, there exists a symmetric matrix $H$ such that $0 \prec H \preceq \beta \cdot I$ and $y = Hx$.
\end{lemma}
\begin{proof}
Set $H = \beta I + \frac{1}{(y-\beta x)^\top x} (y-\beta x)(y-\beta x)^\top$. The matrix $H$ is well-defined. Indeed, from the assumption $\beta x^\top y > \|y\|^2$, we obtain that $\|y\| < \beta \|x\|$. Consequently, $y^\top x \less \|y\|\|x\| < \beta \|x\|^2$ so that the term $(y-\beta x)^\top x$ in the denominator is negative. This implies in particular that $H \preceq \beta I$. A simple calculation yields that $Hx = y$. It remains to show that $H \succ 0$. This holds provided that $\beta > \|y-\beta x\|^2 / (\beta \|x\|^2 - x^\top y)$, i.e., $\|y\|^2 < \beta x^\top y$ which is true by assumption.
\end{proof}
\begin{lemma}
\label{lemmasymmetricmappingtogradients}
Suppose that the function $f$ is $\gamma$-strongly convex and $\mu$-strongly smooth. Let $z^*$ be the solution to the dual program, and $y^*$ the solution to the sketched dual program. Let $g_{z^*} \in \partial f^*(z^*)$ and $g_{y^*} \in \partial f^*(y^*)$. Then, there exists a symmetric matrix $0 \prec H \preceq \frac{2}{\gamma}\cdot I$ such that $g_{y^*} - g_{z^*} = H (y^* - z^*)$.
\end{lemma}
\begin{proof}
From standard convex analysis arguments, the function $f^*$ is $(1/\gamma)$-smooth: this implies that $\|g_{y^*}-g_{z^*}\|_2^2 \less \frac{1}{\gamma} (g_{y^*} - g_{z^*})^\top (y^* - z^*)$. The function $f^*$ is also $(1/\mu)$-strongly convex: this implies that $\frac{1}{\mu}\|y^*-z^*\|^2_2 \less (g_{y^*}-g_{z^*})^\top (y^*-z^*)$. If $y^* = z^*$ then we obtain from the former inequality that $g_{y^*} = g_{z^*}$ and the matrix $H = \frac{2}{\gamma}\cdot I$ trivially satisfies the claim. Hence, we suppose that $y^* \neq z^*$. From the latter inequality, we have that $(g_{y^*}-g_{z^*})^\top (y^*-z^*) > 0$. Combining this with the former inequality, we obtain the strict inequality $\|g_{y^*}-g_{z^*}\|_2^2 < \frac{2}{\gamma} (g_{y^*} - g_{z^*})^\top (y^* - z^*)$. Then, the claim immediately follows from Lemma~\ref{lemmaexistencesymmetricmappingsmoothness}.
\end{proof}

\begin{lemma}[Residuals of random projections]
\label{lemmadecompositionrandomprojection}
Let $x \in \real^d$ be a given vector such that $\|x\|_2=1$. Let $Q \in \real^{d \times m}$ be a partial Haar matrix in $\real^d$. Consider the matrix $P^\perp = I - Q Q^\top$. Then, we have the decomposition $\frac{P^\perp x}{\|P^\perp x\|_2} = \alpha \, x + \sqrt{1-\alpha^2} \,X Z$, where $\alpha = \|P^\perp x\|_2$, $X \in \real^{d \times (d-1)}$ is an orthonormal complement to $x$ (i.e., $[x, X]$ is an orthogonal matrix), and $Z$ is a $(d-1)$-dimensional vector uniformly distributed onto the unit sphere $\mathcal{S}^{d-1} \defn \{w \in \real^d \mid \|w\|_2=1\}$ and independent of $\alpha$.
\end{lemma}
\begin{proof}
We use the orthogonal decomposition $Q Q^\top x = \begin{bmatrix} x & X \end{bmatrix}\cdot\begin{bmatrix} x^\top \\ X^\top\end{bmatrix} Q Q^\top x = \|Q^\top x\|_2^2 \, x + X X^\top Q Q^\top x$, so that $\frac{P^\perp x}{\|P^\perp x\|_2} = \alpha \, x + X \wtilde Z$ where $\wtilde Z \defn -\frac{X^\top Q Q^\top x}{\|P^\perp x\|_2}$. The vectors $\alpha x$ and $X\wtilde Z$ are orthogonal. Taking norms and using that $\|X\wtilde Z\|_2 = \|\wtilde Z\|_2$, we obtain that $\|\wtilde Z\|_2 = \sqrt{1 - \alpha^2}$. Setting $Z = \frac{\wtilde Z}{\sqrt{1-\alpha^2}}$, we obtain that $\|Z\|_2 = 1$ and the decomposition $\frac{P^\perp x}{\|P^\perp x\|_2} = \alpha \, x + \sqrt{1-\alpha^2} \, XZ$. It remains to show that $Z$ is uniformly distributed on the sphere $\mathcal{S}^{d-1}$ and independent of $\alpha$. Let $\Omega \in \real^{(d-1)\times(d-1)}$ be a Haar matrix independent of $Q$. Rotational invariance in distribution of the partial Haar matrix $Q$ implies that $Q Q^\top \overset{\mathrm{d}}{=} \begin{bmatrix} x^\top \\ X^\top \end{bmatrix} Q Q^\top \begin{bmatrix} x & X \end{bmatrix}$, and furthermore,
\begin{align*}
    \begin{bmatrix} 1 & 0 \\ 0 & \Omega \end{bmatrix} \begin{bmatrix} x^\top \\ X^\top \end{bmatrix} Q Q^\top \begin{bmatrix} x & X \end{bmatrix} \begin{bmatrix} 1 & 0 \\ 0 & \Omega \end{bmatrix} \overset{\mathrm{d}}{=} \begin{bmatrix} x^\top \\ X^\top \end{bmatrix} Q Q^\top \begin{bmatrix} x & X \end{bmatrix}\,,
\end{align*}
i.e.,
\begin{align*} 
    \begin{bmatrix} x^\top Q Q^\top x & x^\top Q Q^\top X \Omega \\ \Omega X^\top QQ^\top x & \Omega X^\top Q Q^\top X \Omega\end{bmatrix}  \overset{\mathrm{d}}{=} \begin{bmatrix} x^\top Q Q^\top x & x^\top Q Q^\top X \\X^\top QQ^\top x & X^\top Q Q^\top X\end{bmatrix}\,.
\end{align*}
Consequently, the joint random variable $(\Omega X^\top Q Q^\top x, x^\top Q Q^\top x)$ has the same distribution as $(X^\top Q Q^\top x, x^\top Q Q^\top x)$, i.e., $(\alpha \sqrt{1-\alpha^2}\cdot\Omega Z, 1-\alpha^2) \overset{\mathrm{d}}{=} (\alpha \sqrt{1-\alpha^2} \cdot Z, 1-\alpha^2)$, which further implies that the distribution of $\Omega Z$ conditional on $\alpha$ is equal to the distribution of $Z$ conditional on $\alpha$. The vector $\Omega Z$ is uniformly distributed on the unit sphere and independent of $\alpha$, which implies the same for $Z$. 
\end{proof}

\subsubsection{Proof of Theorem~\ref{theoremlowerboundfirstorderobliviousrandom}}

According to Propositions~\ref{propositionfenchelduality} and~\ref{propositionsketchedfenchelduality}, there exist $g_{z^*} \in \partial f^*(z^*)$ and $g_{y^*} \in \partial f^*(y^*)$ such that $g_{z^*} + \frac{1}{\lambda} A A^\top z^* = 0$ and $g_{y^*} + \frac{1}{\lambda} A P_S A^\top y^* = 0$. Subtracting the previous two equalities and multiplying by $\lambda$, we obtain that $\lambda(g_{y^*} - g_{z^*}) + A P_S A^\top (y^* - z^*) = A P_S^\perp A^\top z^*$. Using the assumption that $f$ is $\gamma$-strongly convex, it follows from Lemma~\ref{lemmasymmetricmappingtogradients} that there exists a symmetric matrix $H$ such that $0 \prec H \preceq \frac{2}{\gamma} I$ and $H (y^* - z^*) = g_{y^*} - g_{z^*}$. Substituting the latter equality into the former, left-multiplying both sides of the resulting equation by $\frac{1}{\lambda} A^\top (\lambda H + A P_S A^\top)^{-1}$ and using the identities $x^* = -\lambda^{-1} A^\top z^*$ and $\xone = -\lambda^{-1} A^\top y^*$, we obtain that $\xone - x^* = A^\top (\lambda H + A P_S A^\top)^{-1} AP_S^\perp x^*$. Multiplying both sides by ${x^*}^\top P_S^\perp$ and using that $\langle P_S^\perp x^*, \xone - x^* \rangle \less \|P_S^\perp x^*\|_2 \|\xone - x^*\|$, it follows that
\begin{align*}
    \frac{\|\xone - x^*\|_2}{\|x^*\|_2} &\gre \frac{\|P_S^\perp x^*\|_2}{\|x^*\|_2} \big\|\big(\lambda H + A P_S  A^\top\big)^{-\frac{1}{2}} A \frac{P_S^\perp x^*}{\|P_S^\perp x^*\|_2}\big\|^2_2\\
    & \underset{(i)}{\gre} \frac{\|P_S^\perp x^*\|_2}{\|x^*\|_2} \big\|\underbrace{\big(2\lambda \gamma^{-1} I_n + A A^\top\big)^{-\frac{1}{2}} A}_{\defn\, M} \frac{P_S^\perp x^*}{\|P_S^\perp x^*\|_2}\big\|^2_2\,.
\end{align*}
\noindent In inequality (i), we used the fact that $A^\top (\lambda H + A P_S A^\top)^{-1} A \succeq M^\top M$. According to Lemma~\ref{lemmadecompositionrandomprojection}, we have that $\frac{P_S^\perp x^*}{\|P_S^\perp x^*\|_2} \overset{\textrm{d}}{=} \alpha \frac{x^*}{\|x^*\|_2} + \sqrt{1-\alpha^2} \cdot X Z$, where $\alpha \defn \|P_S^\perp \frac{x^*}{\|x^*\|_2}\|_2$, $X \in \real^{d \times (d-1)}$ is an orthonormal complement to $\frac{x^*}{\|x^*\|_2}$ and $Z \in \real^{d-1}$ is a random vector which is uniformly distributed onto the unit sphere in $\real^{d-1}$ and independent of $\alpha$. Consequently, we have that
\begin{align*}
    \mathbb{E}_S \frac{\|\xone - x^*\|_2}{\|x^*\|_2} &\gre \mathbb{E}_S \|\alpha^{\frac{3}{2}} M \frac{x^*}{\|x^*\|} + \sqrt{\alpha (1-\alpha^2)} MXZ \|_2^2\\
    &\underset{(i)}{=} \mathbb{E}_S\{\alpha^3\} \cdot \frac{\|Mx^*\|_2^2}{\|x^*\|_2^2} + \underbrace{\mathbb{E}_S\{\alpha (1-\alpha^2)\} \cdot \mathbb{E}\|MXZ\|_2^2}_{\gre 0}\\
    & \underset{(ii)}{\gre} \mathbb{E}_S\{\alpha^2\}^\frac{3}{2} \cdot \frac{\|Mx^*\|_2^2}{\|x^*\|_2^2}\,.
\end{align*}
In inequality (i), we used that the cross-term in the expansion of the square is equal to $0$ because of the independence of $\alpha$ and $Z$ and the fact that $\mathbb{E}Z = 0$. In inequality (ii), we used Jensen's inequality to obtain that $\mathbb{E}_S\{\alpha^2\}^{\frac{3}{2}} \less \mathbb{E}_S\{\alpha^3\}$. Using that $\mathbb{E}_S\{\alpha^2\} = 1-\frac{m}{d}$, $\|M\|_2^2 = \frac{\sigma_1^2}{\sigma_1^2 + \frac{2\lambda}{\gamma}}$ and taking the supremum over $f \in \mathcal{F}_{\gamma,\mu}$, we obtain that 
\begin{align*}
    \sup_{f \in \mathcal{F}_{\gamma,\mu}} \,\mathbb{E}_S\!\left\{\frac{\|\xone - x^*\|_2^2}{\|x^*\|_2^2}\right\} \underset{(i)}{\gre} \sup_{f \in \mathcal{F}_{\gamma,\mu}} \,\mathbb{E}_S\!\left\{\frac{\|\xone - x^*\|_2}{\|x^*\|_2}\right\}^2 \gre \left(1-\frac{m}{d}\right)^3 \cdot \frac{\sigma_1^4}{(\sigma_1^2 + \frac{2\lambda}{\gamma})^2}\,,
\end{align*}
where inequality (i) is a consequence of Jensen's inequality. This concludes the proof.
\qed

\section{Right-sketching and statistical optimality}

\subsection{Technical preliminaries}

Given a radius $\delta > 0$, we say that $\{x_1, \dots, x_M\} \subset \real^d$ is a $\delta$-packing of $\mathcal{B}_2^d$ in the metric $\|P_{AS}A \cdot\|_2$ if $x_j \in \mathcal{B}_2^d$ for any $j \in \{1,\dots,M\}$, and, $\|P_{AS}A (x_j - x_k)\|_2 > \delta$ for any $j \neq k$. We say that the $\delta$-packing $\{x_1, \dots, x_M\}$ is maximal if for any $x \in \mathcal{B}_2^d$, there exists $i \in \{1,\dots,M\}$ such that $\|P_{AS}A(x - x_i)\|_2 \less \delta$. We recall that we denote the rank of the matrix $A$ by $\rho$.

\begin{lemma}
\label{lemmapackingnumber}
Let $K \in \{1,\dots,\rho\}$ be any index such that $\sigma^2_{K} > \|P_{AS}^\perp A\|_2^2$, define $\delta_K \defn \sqrt{\sigma_{K}^2 - \|P_{AS}^\perp A\|_2^2}$ and let $\delta \in (0,\delta_K)$. Then, there exists a $\frac{\delta}{2}$-packing $\{x_j\}_{j=1}^M$ of $\mathcal{B}_2^d$ in the metric $\|P_{AS}A\cdot\|_{2}$ such that $\log M \gre K\cdot \log 2$, and such that $\|P_{AS}A (x_j-x_k)\|_2 \less 2\delta$ for all $j,k \in \{1,\dots,M\}$.
\end{lemma}
\begin{proof}
Let $\what U \what \Sigma \what V^\top$ be a thin SVD of $P_{AS}A$ and denote its rank by $\what \rho$. Similarly, let $U \Sigma V^\top$ be a thin SVD of $A$. Let $\what \sigma_1 \gre \dots \gre \what \sigma_{\what \rho} > 0$ (resp.~$\sigma_1 \gre \dots \gre \sigma_\rho > 0$) be the singular values of $P_{AS}A$ (resp.~$A$). It holds that $|\sigma^2_j - \what \sigma_j^2| \less \|P_{AS}^\perp A\|_2^2$. Consider the Euclidean ball $\mathcal{B}^K_2\!(\delta^2) \defn \left\{\theta \in \real^K \mid \sum_{j=1}^K \frac{\theta_j^2}{\delta^2} \less 1\right\}$. For any $\theta \in \mathbb{R}^K$, we have $\sum_{j=1}^K \frac{\theta_j^2}{\delta^2} \, \underset{(i)}{\gre} \, \sum_{j=1}^K \frac{\theta_j^2}{\sigma_j^2-\|P_{AS}^\perp A\|_2^2} \,\underset{(ii)}{\gre} \, \sum_{j=1}^K \frac{\theta_j^2}{\what \sigma_j^2}$, where inequality (i) follows from $\delta^2 \less \sigma_j^2-\|P_{AS}^\perp A\|_2^2$ for all $j=1,\dots,K$ and inequality (ii) follows from the fact that $|\sigma^2_j - \what \sigma_j^2| \less \|P_{AS}^\perp A\|_2^2$ for all $j=1,\dots,K$. Therefore, if $\theta \in \mathcal{B}^K_2\!(\delta^2)$, then $\wtilde \theta \defn [\theta_1, \dots, \theta_K, 0,\dots, 0]$ belongs to the ellipsoid $\mathcal{E}^{\what \rho}_{\what \sigma} \defn \left\{\overline \theta \in \real^{\what \rho} \mid \sum_{j=1}^{\what \rho} \frac{\overline{\theta}^2_j}{\what \sigma_j^2} \less 1\right\}$. Consequently, if $\{\theta^a\}_{a=1}^M$ is a $\frac{\delta}{2}$-packing of $\mathcal{B}^K_2\!(\delta^2)$ in the metric $\|\cdot\|_2$, then $\{\wtilde{\theta}^a\}_{a=1}^M$ is a $\frac{\delta}{2}$-packing of $\mathcal{E}^{\what \rho}_{\what \sigma}$ in the metric $\|\cdot\|_2$. It is well-known that there exists such a packing $\{\theta^a\}_{a=1}^M$ with cardinality $M \gre 2^K$. Hence, there exists a $\frac{\delta}{2}$-packing $\{\wtilde{\theta}^a\}_{a=1}^M$ of $\mathcal{E}^{\what \rho}_{\what \sigma}$ in the metric $\|\cdot\|_2$ with cardinality $M \gre 2^K$. 

We are now ready to construct the claimed packing of $\mathcal{B}_2^d$ in the metric $\|P_{AS}A\cdot\|_{2}$. For each $a=1,\dots,M$, we set $x_a \defn \what V \what \Sigma^{-1} \wtilde{\theta}^a$. Observe that $\|x_a\|_2^2 = \|\what \Sigma^{-1} \wtilde{\theta}^a \|_2^2 \less 1$, i.e., $x_a \in \mathcal{B}_2^d$. For $a \neq b$, we have $\|P_{AS}A (x_a-x_b)\|_2 = \|\what U \what \Sigma \what V^\top \what V \what \Sigma^{-1} (\wtilde{\theta}^a - \wtilde{\theta}^b)\|_2 = \|\wtilde{\theta}^a - \wtilde{\theta}^b\|_2 \gre \frac{\delta}{2}$. Furthermore, we have by construction that $\|\wtilde \theta^a\|_2 \less \delta$, which implies that $\|P_{AS}A(x_a - x_b)\|_2 \less 2 \delta$, and this concludes the proof.
\end{proof}

\subsection{Proof of Theorem~\ref{theoremlowerboundleastsquares}}

Our proof is based on the standard local Fano method. Given an arbitrary vector $\bar x \in \real^d$, we denote by $\mathbb{P}_{\bar{x}}$ the probability measure with respect to the distribution $\mathcal{N}(A\bar x, \frac{\sigma^2}{n}I_n)$ and by $\mathbb{E}_{\bar x}$ the corresponding expectation. Fix a sketch size $m \gre 1$, an embedding $S \in \real^{d\times m}$ and let us assume that $\|P_{AS}^\perp A\|_2^2 \less \frac{\sigma_{d_s+1}^2}{2}$. Using that $\|P_{AS}\|_2 \less 1$, we have $\mathfrak{M}_S \gre \inf_{\what x} \sup_{\|x_\text{pl}\|_2 \less 1} \mathbb{E}_{x_\text{pl}}\|P_{AS}A(\what x - x_\text{pl})\|_2^2$. It is thus sufficient to lower bound the latter quantity. We denote $b_S \defn S^\top A^\top b$, and we let $\what x \equiv \what x (b_S)$ be an estimator. We introduce the radius $\delta_m \defn \sqrt{\sigma_{d_s+1}^2 - \|P_{AS}^\perp\|_2^2}$. By assumption, we have $\delta_m \gre \frac{\sigma_{d_s+1}}{\sqrt{2}}$. Using Lemma~\ref{lemmapackingnumber} with $K=d_s+1$ and $\delta = \frac{\delta_m}{8}$, we obtain that there exists a maximal $\frac{\delta_m}{16}$-packing $\{x_1, \dots, x_M\}$ of $\mathcal{B}_2^d$ in the metric $\|P_{AS}A \cdot\|_2$ such that $\log M \gre d_s \cdot \log 2$, and, $\|P_{AS}A(x_j - x_k)\|_2 \less \frac{\delta_m}{4}$ for all $j,k$. Then, we have
\begin{align*}
    \sup_{\|x_\text{pl}\|_2 \less 1} \Exs_{x_\text{pl}} \|P_{AS} A(\what x - x_\text{pl})\|_2^2 \gre \frac{1}{M} \sum_{j=1}^M \Exs_{x_j} \|P_{AS} A(\what x - x_j)\|_2^2 \underset{(i)}{\gre} \frac{\delta_m^2}{32^2} \cdot \frac{1}{M} \sum_{j=1}^M \mathbb{P}_{x_j}\!\left(\|P_{AS} A (\what x - x_j)\|_2 \gre \frac{\delta_m}{32} \right)\,,
\end{align*}
where inequality (i) follows from Markov's inequality. We introduce the test function $\psi(\what x) \defn \argmin_{k=1,\dots,M} \|P_{AS} A(\what x - x_k)\|_2$. We claim that $\|P_{AS} A (\what x - x_j)\|_2 < \frac{\delta_m}{32}$ implies that $\psi(\what x) = j$. Indeed, suppose that $\|P_{AS} A (\what x - x_j)\|_2 < \frac{\delta_m}{32}$. Then, using the fact that $\{x_i\}_{i=1}^M$ is a $(\delta_m/16)$-packing, we have for any $k \neq j$ that $\|P_{AS} A (\what x - x_k)\|_2 \gre \|P_{AS} A (x_j - x_k)\|_2 - \|P_{AS} A (\what x - x_j)\|_2 > \delta_m/16 - \delta_m/32 = \delta_m/32$, i.e., $\psi(\what x) = j$. It follows that
\begin{align*}
    \sup_{\|x\|_\text{pl} \less 1} \mathbb{E}_{x_\text{pl}} \|P_{AS}A(\what x - x_\text{pl})\|_2^2 \gre \left(\frac{\delta}{32}\right)^2 \cdot \frac{1}{M} \sum_{j=1}^M \mathbb{P}_{x_j}(\psi(\what x) \neq j) \gre \left(\frac{\delta}{32}\right)^2 \cdot \inf_{\psi} \mathbb{E}_J\! \left\{\mathbb{P}_{x_j}(\psi(b_S) \neq j) \mid J = j \right\}\,,
\end{align*}
where $J$ is a uniformly distributed random variable in $\{1, \dots, M\}$. Fano's inequality states that the testing error in the latter right-hand side is lower bounded by the quantity $(1 - \frac{I(b_S; J) + \log 2}{\log M})$, where $I(b_S;J)$ denotes the mutual information between the random variables $b_S$ and $J$. Consequently, we obtain
\begin{align}
\label{eqnintermediatefano}
    \sup_{\|x\|_\text{pl} \less 1} \mathbb{E}_{x_\text{pl}} \|P_{AS}A(\what x - x_\text{pl})\|_2^2 \gre \left(\frac{\delta}{32}\right)^2 \cdot \left(1 - \frac{I(b_S; J) + \log 2}{\log M}\right)\,.
\end{align}
Introducing the mixture distribution $P_{x_J} \defn \frac{1}{M} \sum_{j=1}^M P_{x_j}$, denoting the Kullback-Leibler (KL) divergence between two distributions $P$ and $Q$ by $D_{\text{kl}}\!\left(P \,\|\,Q\right)$ and using the convexity of $Q \mapsto D_{\text{kl}}\!\left(P \,\|\,Q\right)$, we have 
\begin{align*}
    I(b_S; J) = \frac{1}{M} \sum_{j=1}^M D_{\text{kl}}\!\left(P_{x_j} \,\|\,P_{x_J}\right) \less \frac{1}{M^2} \sum_{j,k=1}^M D_{\text{kl}}\!\left(P_{x_j} \,\|\,P_{x_k}\right)\,.
\end{align*}
Using that the KL divergence between two Gaussian distributions $P_{x_j}$ and $P_{x_k}$ is equal to $\frac{n}{\sigma^2} \|P_{AS}A(x_j-x_k)\|_2^2$, it follows that
\begin{align*}
    I(b_S; J) \less \frac{1}{M^2} \sum_{j,k=1}^M \frac{n}{\sigma^2} \|P_{AS}A(x_j - x_k)\|_2^2 \less \frac{n \delta_m^2}{16 \sigma^2}\,,
\end{align*}
where in the last inequality, we used the fact that $\|P_{AS}A(x_j - x_k)\|_2 \less \frac{\delta_m}{4}$ for any $j,k$. Combining these observations, we obtain from~\eqref{eqnintermediatefano} that
\begin{align*}
    \sup_{\|x\|_\text{pl} \less 1} \mathbb{E}_{x_\text{pl}} \|P_{AS}A(\what x - x_\text{pl})\|_2^2 \gre \frac{\delta_m^2}{32^2} \left(1 - \frac{n\delta_m^2}{16 \sigma^2 d_s \log 2} - \frac{1}{d_s}\right)\,.
\end{align*}
Using that $\frac{\sigma_{d_s+1}^2}{2} \less \delta_m^2 \less \sigma_{d_s+1}^2$ and $\frac{\sigma^2 d_s}{n} \gre \sigma_{d_s+1}^2$, it follows that $\frac{n\delta_m^2}{16 \sigma^2 d_s \log 2} \less \frac{n\sigma_{d_s+1}^2}{16 \sigma^2 d_s \log 2} \less \frac{1}{16 \log 2}$, and thus, 
\begin{align*}
    \sup_{\|x\|_\text{pl} \less 1} \mathbb{E}_{x_\text{pl}} \|P_{AS}A(\what x - x_\text{pl})\|_2^2 \gre \frac{\sigma^2_{d_s+1}}{2 \cdot 32^2} \cdot \underbrace{\left(1-\frac{1}{16 \log 2} - \frac{1}{d_s}\right)}_{\gre \frac{1}{2}}\,,
\end{align*}
and this concludes the proof of the lower bound. 

\noindent \textbf{Polynomial decay.} Consider a polynomial decay $\sigma_j = j^{-\frac{1+\nu}{2}}$ for some $\nu > 0$. The scaling relation $\frac{\sigma^2 d_s}{n} \asymp \sigma_{d_s+1}^2$ yields that $d_s \asymp \left(\frac{n}{\sigma^2}\right)^{\frac{1}{2+\nu}}$. Fix a target rank $k \gre 1$ and a sketch size $m=2k$. According to Lemma~\ref{lemmaupperboundZf} and using the inequality $(a+b)^2 \less 2 a^2 + 2 b^2$ for any $a,b\in \real$, we have with probability at least $1-6e^{-k}$ that 
\begin{align*}
    \|P_{AS}^\perp A\|_2^2 \less 26^2 \cdot \left(\sigma_{k+1} + \frac{1}{\sqrt{k}} \sqrt{\sum_{j=k+1}^\rho \sigma_j^2}\right)^2 \less 2 \cdot 26^2 \cdot (\sigma_{k+1}^2 + \frac{1}{k} \sum_{j=k+1}^\rho \sigma_j^2)\,.
\end{align*}
We have $\frac{1}{k} \sum_{j=k+1}^\rho \sigma_j^2 \less \frac{1}{k} \int_{k}^{+\infty} u^{-(1+\nu)} \mathrm{d}u = \frac{1}{\nu k^{1+\nu}}$. It follows that $\|P_{AS}^\perp A\|_2^2 \less 2 \cdot 26^2 \cdot (\nu^{-1}+1) \cdot k^{-(1+\nu)}$ with probability at least $1-6e^{-k}$. Consequently, we have $\|P_{AS}^\perp A\|_2^2 \less \frac{\sigma_{d_s+1}^2}{2}$ if $(d_s+1)^{-(1+\nu)} \gre 2704 \cdot (\nu^{-1}+1) k^{-(1+\nu)}$, for which it is sufficient to have $m \gre \underbrace{(\ceil{2 \cdot (2704(\nu^{-1}+1))^{(1+\nu)}}+1)}_{\defn \, c_\nu^\text{poly}} \cdot d_s$.\\
\\
\noindent \textbf{Exponential decay.} Consider an exponential decay $\sigma_j = e^{-\frac{\nu j}{2}}$ for some $\nu > 0$. The scaling relation $\frac{\sigma^2 d_s}{n} \asymp \sigma_{d_s+1}^2$ yields that $d_s \asymp \frac{1}{\nu} \cdot \log(n/\sigma^2)$. Fix a target rank $k \gre 1$ and a sketch size $m=2k$. We have with probability at least $1-6e^{-k}$ that $\|P_{AS}^\perp A\|_2^2 \less 2 \cdot 26^2 \cdot (\sigma_{k+1}^2 + \frac{1}{k} \sum_{j=k+1}^\rho \sigma_j^2) \less \frac{2704}{\nu}\cdot e^{-(k+1)\nu}$. Consequently, $\|P_{AS}^\perp A\|_2^2 \less \frac{\sigma_{d_s+1}^2}{2}$ if $\frac{5408}{\nu} e^{-(k+1)\nu} \less e^{-(d_s+1)\nu}$, i.e., $m \gre 2 \cdot d_s + \frac{2}{\nu} \log(5408/\nu)$. Assuming that $\log(n/\sigma^2) \gre \log(5408/\nu)$ and using that $d_s \asymp \frac{1}{\nu}\cdot \log(n/\sigma^2)$, we get that $\|P_{AS}^\perp A\|_2^2 \less \frac{\sigma_{d_s+1}^2}{2}$ with probability at least $1-6e^{-d_s}$ for $m \gre 4d_s$.

\subsection{Proof of Corollary~\ref{corollaryoptimalityzeroorder}}

Through a simple calculation, we obtain that the risk of $\xzero$, as $\lambda \to 0$, satisfies the bias-variance decomposition
\begin{align}
\label{eqnbiasvariancedecomposition}
    \lim_{\lambda \to 0} \mathfrak{R}(\xzero) = \underbrace{\mathbb{E}\|P_{AS}w\|_2^2}_{=\,m \sigma^2 / n} + \sup_{\|x_\text{pl}\|_2 \less 1} \|P_{AS}^\perp A x_\text{pl}\|_2^2 = \frac{\sigma^2 m }{n} + \|P_{AS}^\perp A\|_2^2\,.
\end{align}
According to~\eqref{eqnupperboundresidualgaussian}, the residual error verifies $\|P_{AS}^\perp A\|_2^2 \lesssim R^2_{m/2}(A)$, so that
\begin{align}
    \mathfrak{M}_S = \inf_{\what x} \mathfrak{R}(\what x) \less \inf_m \lim_{\lambda \to 0} \mathfrak{R}(\xzero) \lesssim \inf_m \left\{ \frac{\sigma^2 m}{n} + R_{m/2}^2(A) \right\}\,.
\end{align}
Consequently, the sketch size $m$ controls this bias-variance trade-off: the larger the sketch size $m$, the larger the variance term $\frac{\sigma^2 m}{n}$ and the smaller the bias $R^2_{m/2}(A)$, and vice-versa. Furthermore, according to Theorem~\ref{theoremlowerboundleastsquares}, under the event $\|P_{AS}^\perp A\|_2^2 \less \frac{\sigma_{d_s+1}^2}{2}$, it holds that $\mathfrak{M}_S \gre c_0 \cdot \sigma_{d_s+1}^2$.\\
\\
\noindent \textbf{Polynomial decay.} Consider a polynomial decay $\sigma_j = j^{-\frac{1+\nu}{2}}$ for some $\nu > 0$. Picking $m= c^\text{poly}_\nu \cdot d_s$ and following the same steps as in the proof of Theorem~\ref{theoremlowerboundleastsquares}, we obtain that $\|P_{AS}^\perp A\|_2^2 \less \frac{\sigma_{d_s+1}^2}{2}$ with probability at least $1-6e^{-d_s}$. Plugging-in this value of $m$ and this bound on the residual error $\|P_{AS}^\perp A\|_2^2$ into~\eqref{eqnbiasvariancedecomposition}, it follows that
\begin{align*}
    \lim_{\lambda \to 0}\mathfrak{R}(\xzero) \less c^\text{poly}_\nu \cdot \frac{\sigma^2 d_s}{n} + \frac{\sigma_{d_s+1}^2}{2} \less (c^\text{poly}_\nu + \frac{1}{2}) \cdot \frac{\sigma^2 d_s}{n}\,.
\end{align*}
with probability at least $1-6e^{-d_s}$.\\
\\
\noindent \textbf{Exponential decay.} Consider an exponential decay $\sigma_j = e^{-\frac{\nu j}{2}}$ for some $\nu > 0$. Picking $m=4 d_s$ and following the same steps as in the proof of Theorem~\ref{theoremlowerboundleastsquares}, we obtain that $\|P_{AS}^\perp A\|_2^2 \less \frac{\sigma_{d_s+1}^2}{2}$ with probability at least $1-6e^{-d_s}$. Plugging-in this value of $m$ and this bound on the residual error $\|P_{AS}^\perp A\|_2^2$ into~\eqref{eqnbiasvariancedecomposition}, it follows that
\begin{align*}
    \lim_{\lambda \to 0}\mathfrak{R}(\xzero) \less 4 \cdot \frac{\sigma^2 d_s}{n} + \frac{\sigma_{d_s+1}^2}{2} \less 5 \cdot \frac{\sigma^2 d_s}{n}\,.
\end{align*}
with probability at least $1-6e^{-d_s}$.

\section{Proofs of Fenchel duality results}

\subsection{Proof of Proposition~\ref{propositionfenchelduality}}
The primal objective function is strongly convex, so that it admits a unique minimizer $x^*$. According to Corollary 31.2.1 in~\cite{rockafellar2015convex} whose assumptions are trivially satisfied, strong duality holds and there exists a dual solution $z^*$. According to Theorem 31.3 in~\cite{rockafellar2015convex}, we have the KKT conditions $x^* = -A^\top z^* / \lambda$ and $z^* = \nabla f(Ax^*)$. In particular, the relation $z^* = \nabla f(Ax^*)$ along with the uniqueness of $x^*$ imply that $z^*$ is unique. If the function $f$ is strictly convex, according to Theorem 26.5 in~\cite{rockafellar2015convex}, the mapping $\nabla f$ is one-to-one from $\real^n$ to the interior of the domain of $f^*$. Consequently, $\nabla f(Ax^*) \in \textrm{int dom }f^*$, i.e., $z^* \in \textrm{int dom }f^*$.
\qed 

\subsection{Proof of Proposition~\ref{propositionsketchedfenchelduality}}

According to Corollary~31.2.1 in~\cite{rockafellar2015convex} whose assumptions are trivially satisfied, there exists a primal solution $\alpha^* \in \real^*$, strong duality holds, and there exists a sketched dual solution $y^* \in \text{dom}\,f^*$. According to Theorem~31.3 in~\cite{rockafellar2015convex}, we have the KKT conditions $S^\top S\alpha^* = -\frac{S^\top A^\top y^*}{\lambda}$ and $y^* = \nabla f(AS\alpha^*)$. If the function $f$ is strictly convex, according to Theorem 26.5 in~\cite{rockafellar2015convex}, the mapping $\nabla f$ is one-to-one from $\real^n$ to the interior of the domain of $f^*$. Consequently, $\nabla f(AS\alpha^*) \in \textrm{int dom }f^*$, i.e., $y^* \in \textrm{int dom }f^*$.
\qed

\section*{Acknowledgment}

\noindent This work was partially supported by the National Science Foundation under grants IIS-1838179, ECCS-2037304, Facebook Research, Adobe Research and Stanford SystemX Alliance.

\ifCLASSOPTIONcaptionsoff
  \newpage
\fi

\bibliographystyle{IEEEtran}
\bibliography{main}

\end{document}